\documentclass[11pt]{article} 
\usepackage{fullpage}
\usepackage{graphicx}
\usepackage{upref}
\usepackage{enumerate}
\usepackage{latexsym}
\usepackage{pdfpages}
\usepackage[bookmarks]{hyperref}

\usepackage{color,graphics}
\usepackage{comment} 
\usepackage{caption}
\usepackage{subcaption}
\usepackage{wrapfig}
\usepackage{braket}
\usepackage{amssymb}
\usepackage{amsmath}
\usepackage{amsthm}
\usepackage{mathrsfs}
\usepackage{mathtools}
\usepackage{commath}
\usepackage{thmtools,thm-restate}

\usepackage{boxedminipage}

\newcommand{\mR}{\mathbb{R}}
\newcommand{\mZ}{\mathbb{Z}}
\newcommand{\cL}{\mathcal{L}}
\newcommand{\lat}{\mathcal{L}}
\newcommand{\basis}{\mathbf{B}}

\usepackage{amsfonts}
\usepackage{varioref}
\usepackage[ansinew]{inputenc}
\usepackage[many]{tcolorbox}
\usepackage{xcolor}
\hypersetup{
	colorlinks,
	linkcolor={red!75!black},
	citecolor={blue!75!black},
	urlcolor={blue!75!black}
}
\usepackage[hyperpageref]{backref}

\newtheorem{theorem}{Theorem}[section]
\newtheorem{lemma}[theorem]{Lemma}
\newtheorem{corollary}[theorem]{Corollary}
\newtheorem{prop}[theorem]{Proposition}
\newtheorem{claim}[theorem]{Claim}

\newtheorem{defn}[theorem]{Definition}

\theoremstyle{remark}

\theoremstyle{definition}

\AtBeginDocument{}

\DeclareMathOperator{\real}{\mathbb{R}}
\DeclareMathOperator{\nat}{\mathbb{N}}

\newcommand{\intg}{\mathbb{Z}}
\newcommand{\ratn}{\mathbb{Q}}

\newcommand{\poly}{\mathrm{poly}}

\newcommand{\svp}{\textsf{SVP}}
\newcommand{\SVP}{\textsf{SVP}}
\newcommand{\uSVP}{\textsf{uSVP}}
\newcommand{\cvp}{\textsf{CVP}}
\newcommand{\CVP}{\textsf{CVP}}
\newcommand{\dss}{\textsf{DSS}}
\newcommand{\DSS}{\textsf{DSS}}
\newcommand{\BDD}{\textsf{BDD}}
\newcommand{\gappvcp}{\textsf{GapPVCP}}

\newcommand{\cB}{{\mathcal B}}

\newcommand{\cO}{{\mathcal O}}

\newcommand{\cS}{{\mathcal S}}
\newcommand{\cT}{{\mathcal T}}

\newcommand{\vect}[1]{\boldsymbol{#1}}
\newcommand{\eps}{\varepsilon}
\renewcommand{\epsilon}{\varepsilon}
\renewcommand{\vec}[1]{\vect{#1}}
\newcommand{\Z}{\mathbb{Z}}
\newcommand{\R}{\mathbb{R}}

\newcommand{\vol}{\mathrm{Vol}}
\newcommand{\dist}{\mathrm{dist}}
\newcommand{\spn}{\mathrm{span}}
\newcommand{\transpose}[1]{\ensuremath{#1^{\scriptscriptstyle T}}}

\begin{document}

\title{Dimension-Preserving Reductions Between SVP and CVP\\in Different $p$-Norms}
\author{Divesh Aggarwal\\CQT, National University of Singapore \\
\texttt{dcsdiva@nus.edu.sg} \and Yanlin Chen\\
Centrum Wiskunde $\&$ Informatica\\
\texttt{yanlin@cwi.nl}
\and Rajendra Kumar
\\
Indian Institute of Technology, Kanpur\\ and National University of Singapore\\
\texttt{rjndr2503@gmail.com}
\and Zeyong Li\\CQT, National University of Singapore \\ \texttt{li.zeyong@u.nus.edu} \and Noah Stephens-Davidowitz\\
Cornell University\\
\texttt{noahsd@gmail.com}}
\date{}
\maketitle

\begin{abstract}
We show a number of reductions between the Shortest Vector Problem and the Closest Vector Problem over lattices in different $\ell_p$ norms ($\SVP_p$ and $\CVP_p$ respectively).  Specifically, we present the following $2^{\eps m}$-time reductions for $1 \leq p \leq q \leq \infty$, which all increase the rank $n$ and dimension $m$ of the input lattice by at most one:
\begin{itemize}
    \item a reduction from $\widetilde{O}(1/\eps^{1/p})\gamma$-approximate $\SVP_q$ to $\gamma$-approximate $\SVP_p$;
    \item a reduction from $\widetilde{O}(1/\eps^{1/p}) \gamma$-approximate $\CVP_p$ to $\gamma$-approximate $\CVP_q$; and
    \item a reduction from $\widetilde{O}(1/\eps^{1+1/p})$-$\CVP_q$ to $(1+\eps)$-unique $\SVP_p$ (which in turn trivially reduces to $(1+\eps)$-approximate $\SVP_p$).
\end{itemize}

The last reduction is interesting even in the case $p = q$. In particular, this special case subsumes much prior work adapting $2^{O(m)}$-time $\SVP_p$ algorithms to solve $O(1)$-approximate $\CVP_p$. In the (important) special case when $p = q$, $1 \leq p \leq 2$, and the $\SVP_p$ oracle is exact, we show a stronger reduction, from $O(1/\eps^{1/p})\text{-}\CVP_p$ to (exact) $\SVP_p$ in $2^{\eps m}$ time. For example, taking $\eps = \log m/m$  and $p = 2$ gives a slight improvement over Kannan's celebrated polynomial-time reduction from $\sqrt{m}\text{-}\CVP_2$ to $\SVP_2$. We also note that the last two reductions can be combined to give a reduction from approximate-$\CVP_p$ to $\SVP_q$ for \emph{any} $p$ and $q$, regardless of whether $p \leq q$ or $p > q$.

Our techniques combine those from the recent breakthrough work of Eisenbrand and Venzin~\cite{EV20} (which showed how to adapt the current fastest known algorithm for these problems in the $\ell_2$ norm to all $\ell_p$ norms) together with sparsification-based techniques.
\end{abstract}

\thispagestyle{empty}
\newpage
\tableofcontents
\pagenumbering{roman}
\newpage
\pagenumbering{arabic}

\section{Introduction}

A lattice $\cL = \cL(\vect{b}_1, \ldots, \vect{b}_n) := \{\sum_{i=1}^n z_i \vect{b}_i : z_i \in \mZ\}$ is the set of all integer linear combinations of linearly independent vectors $\vect{b}_1,\dots,\vect{b}_n \in \mR^m$.
 We call $n$ the rank of the lattice, $m$ the dimension or ambient dimension, and $(\vect{b}_1, \ldots, \vect{b}_n)$ a basis of the lattice.
 
The two most important computational problem on lattices are the Shortest Vector Problem ($\SVP$) and the Closest Vector Problem ($\CVP$). 
Given a basis for a lattice $\cL \subseteq \mR^m$,
$\SVP$ asks us to compute a non-zero vector in $\cL$ that is as short as possible, and $\CVP$ asks us to find a vector in $\cL$ closest to some target point $\vect{t}\in \mR^m$. 

We define ``short'' and ``close'' here in terms of the $\ell_p$ norm for some $1\leq p \leq \infty$, given by
\[
\|\vect{x}\|_p := (\sum\limits_{i=1}^m\abs{x_i} ^p)^{1/p}
\]
for finite $p$ and 
\[
\|\vect{x}\|_\infty := \max_i \abs{x_i} .
\]
We write $\SVP_p$ and $\CVP_p$ for the respective problems in the $\ell_p$ norm. For any approximation factor $\gamma=\gamma(m,n)\geq 1$, we can also define the approximate version of $\SVP_p$, which asks us to find a non-zero lattice vector whose length is within a factor of $\gamma$ of the minimal possible value, called $\gamma$-$\SVP_p$. Correspondingly, $\gamma$-$\CVP_p$ asks us to find a lattice vector whose distance to the target is within a factor of $\gamma$ of the minimal distance. Both problems are known to be NP-hard (under randomized reductions in the case of SVP) for any $1 \leq p \leq \infty$ and any $\gamma = O(1)$~\cite{Khot05}, and even hard for the nearly polynomial factor $\gamma = m^{c/\log \log m}$.\footnote{For $\gamma = m^{c/\log \log m}$, $\CVP_p$ and $\SVP_\infty$ are known to be NP-hard~\cite{DinApproximatingSVPinfty02,DKRSApproximatingCVP03}, while $\SVP_p$ for finite $p$ is only known to be hard under subexponential-time reductions~\cite{Khot05,HR07}.}  And, $\CVP$ is no easier than $\SVP$ in a very strong sense: there is an efficient reduction from $\SVP$ to $\CVP$ that exactly preserves the rank, dimension, norm, and approximation factor~\cite{GMSS99}.

Both $\gamma\text{-}\SVP_p$ and $\gamma\text{-}\CVP_p$ are interesting computational problems for a very wide range of approximation factors $1 \leq \gamma \leq 2^n$. Algorithms for these problems have found a remarkable number of applications in algorithmic number theory~\cite{LLL82}, convex optimization~\cite{Lenstra83,Kannan87,FrankT87}, coding theory~\cite{Buda89}, and cryptanalysis~\cite{Shamir84,Brickell84,LagariasO85,Kannan87,Odl98, JS98, NS01}. Over the past two decades, many cryptographic primitives have been constructed with their security based on the worst-case hardness of (variants of) $\CVP_2$ and $\SVP_2$ with approximation factors $\gamma = \poly(n)$ that are polynomial in the dimension (e.g.,~\cite{Ajtai96,MR04,Regev09,Gentry09,BV14,PeiDecadeLattice16}; notice that when $p = 2$ we may assume without loss of generality that $n = m$).
Such cryptosystems have attracted a lot of research interest due to their conjectured resistance to quantum attacks as well as their useful functionality.

Algorithms for $\gamma\text{-}\SVP_2$ are extremely well-studied, with a rich set of algorithmic techniques pioneered by Lenstra, Lenstra, and Lov\'asz~\cite{LLL82}, Babai~\cite{babai86}, Kannan~\cite{Kannan87}, Schnorr~\cite{Schnorr1987AHO}, and Ajtai, Kumar, and Sivakumar~\cite{AKS01}, among others. This has resulted in a rather complicated landscape of algorithmic results. At a high level, our fastest known algorithms for $\SVP_2$ for $1 \leq \gamma \leq \poly(n)$ run in time $2^{C n}$ for some constant $C$, where the constant $C$ depends on $\gamma$. But, the specific constant $C$ and its specific dependence on $\gamma$ is extremely important, particularly for the practical security of modern cryptography. Even a minor improvement in this constant $C$ could render proposed cryptosystems insecure in practice. See~\cite{ALNS19} for a more detailed discussion of the current state of the art and, e.g.,~\cite{APSConcreteHardness15} for an explanation of the relationship between this constant $C$ and the practical security of lattice-based cryptography.

Until very recently, much less was known about $\gamma\text{-}\SVP_p$ and $\gamma\text{-}\CVP_p$ for $p \neq 2$. The fastest known algorithms for these problems run in $\min\{2^{Cm}, n^{Cn}\}$ time for $\gamma\text{-}\CVP_p$ for $1 + \Omega(1) \leq \gamma < m^{|1/p-1/2|}$ \cite{Kannan87,AKS02,BNSamplingMethods07}, while for $\SVP_p$, $2^{Cn}$-time algorithms are known. But, in both cases the constants in the exponent were not very well studied. (There was some complexity-theoretic evidence suggesting a lower bound of $2^{(1-\eps)n}$-time for small constant approximation factors $\gamma$~\cite{BGS17,AS18b,ABGS19} for $\CVP_p$ and $\SVP_\infty$, but little work on upper bounds for the constant in the exponent.) This potential gap in our knowledge is unfortunate, since for many applications (including some cryptanalytic applications) a new faster algorithm for $\gamma\text{-}\SVP_p$ or $\gamma\text{-}\CVP_p$ with $p \neq 2$ could be just as devastating as a new faster algorithm for $\gamma\text{-}\SVP_2$. (E.g., many practical cryptographic constructions actually work directly with the $\ell_\infty$ or $\ell_1$ norm.)

The relationship between these various problems was also not particularly well understood. Of course, since the $\ell_p$ norms satisfy $\|\vec{x}\|_q \leq \|\vec{x}\|_p \leq m^{1/p-1/q}\|\vec{x}\|_q$ for $p \leq q$, there is a trivial reduction from $(m^{|1/p-1/q|}\gamma)\text{-}\SVP_p$ to $\gamma \text{-}\SVP_q$ for any $p,q \in [1,\infty]$, and likewise for $\CVP$. In particular, this reduction preserves both the dimension and rank of the lattice (since the reduction simply passes its input to its oracle unchanged). 

More interestingly, as we mentioned above, all of these problems are known to be NP-complete (under randomized reductions in the case of SVP) for any constant $\gamma$. So, in a certain very weak sense, they are all equivalent problems when the approximation factor is constant. But, the reductions implied by these completeness results increase the rank and dimension by a large polynomial factor, so that they tell us very little in the context of $2^{\Omega(m)}$-time algorithms. And, they only apply for constant $\gamma$ (or for $\gamma \leq n^{c/\log \log n}$ in the case of $\CVP_p$ and $\SVP_\infty$).\footnote{Ajtai, Kumar, and Sivakumar also showed a $2^{O(n)}$-time reduction that only increased the dimension and rank by one from $O(1)\text{-}\CVP_2$ to a non-standard problem related to $\SVP_2$ (which was later extended to a $2^{O(m)}$-time reduction for all $p$). Specifically, they showed a reduction to the problem of sampling (roughly) uniformly from the set of lattice points in a ball. We will use similar techniques to prove the results described below.}

Regev and Rosen~\cite{RR06} improved substantially on this, by showing how to use norm embeddings to efficiently reduce $C^* \gamma\text{-}\SVP_2$ to $\gamma\text{-}\SVP_p$ for any $p$ and any constant $C^* > 1$, and likewise for $\CVP$. So, in some sense $\ell_2$ is ``the easiest norm.'' However, their reduction increases the dimension substantially---by a factor of $(C^*+1)^2/(C^*-1)^2$ for $p < 2$, $(C^* + 1)^p/(C^*-1)^p \cdot (n/p)^{p/2-1}$ for $2 < p < \infty$ and $n^{O(1/(C^*-1))}$ for $p = \infty$. This blowup in the dimension is particularly significant when $p > 2$, since a superconstant increase in the dimension $m$ is very expensive in the context of $2^{\Omega(m)}$-time algorithms. Furthermore, Regev and Rosen are only able to reduce from the $\ell_2$ norm to other $\ell_p$ norms---i.e., from problems with many known algorithmic techniques to problems with fewer known techniques---which seems less interesting than reductions \emph{from} the $\ell_p$ norm for $p \neq 2$ \emph{to} the $\ell_2$ norm, or more generally between arbitrary $p$ and $q$. (Regev and Rosen's reduction also implied hardness of $\SVP_p$ for some values of $p$ that were not otherwise known to be hard at the time.)

Quite recently, Eisenbrand and Venzin made a major breakthrough in this area by showing that the current fastest known algorithm for $O(1)\text{-}\SVP_2$ can be used as a subprocedure to solve both $O(1)\text{-}\SVP_p$ and $O(1)\text{-}\CVP_p$ for any $p$ in essentially the same running time~\cite{EV20}.\footnote{The fastest known algorithm for $\gamma\text{-}\SVP_2$ for large constant $\gamma$ has a running time that approaches $2^{C_2n+o(n)}$ as $\gamma \to \infty$, where $C_2$ is geometric constant known to satisfy $C_2 < 0.802$~\cite{LWXZShortestLattice11,WLWFindingShortest15,AUV19}. Eisenbrand and Venzin showed the same result for all $\ell_p$ norms in~\cite{EV20}.} This result on its own was quite surprising (at least to the authors of this work) and drastically changed the algorithmic landscape in this area, but the techniques used to achieve the result are perhaps even more surprising. In particular, the Eisenbrand-Venzin algorithms are quite simple. And, though they use a specific and rather technical property of the $\SVP_2$ subprocedure (see the discussion below Theorem~\ref{thm:svp_to_svp_intro}), their algorithms still look suspiciously like dimension- and rank-preserving reductions from $O(1)\text{-}\SVP_p$ and $O(1)\text{-}\CVP_p$ to $O(1)\text{-}\SVP_2$. 

\subsection{Our results and techniques}

Our results are summarized in Figure~\ref{fig:results}. Below, we provide some more details about the results and discuss techniques. As we explain below, many of our techniques are heavily inspired by Eisenbrand and Venzin~\cite{EV20} (though our presentation differs quite a bit from theirs). We also note that all of our reductions are randomized, and they all either preserve the rank and dimension of the input lattice exactly or increase each by exactly one. (We sometimes informally refer to reductions that increase the rank and dimension by one as ``dimension- and rank-preserving,'' since for algorithmic purposes the distinction between maintaining these values exactly and increasing them by one is unimportant.)

\begin{figure}[t]
\begin{center}
\begin{boxedminipage}{0.9 \textwidth}
\begin{center}
\includegraphics[width=0.85 \textwidth]{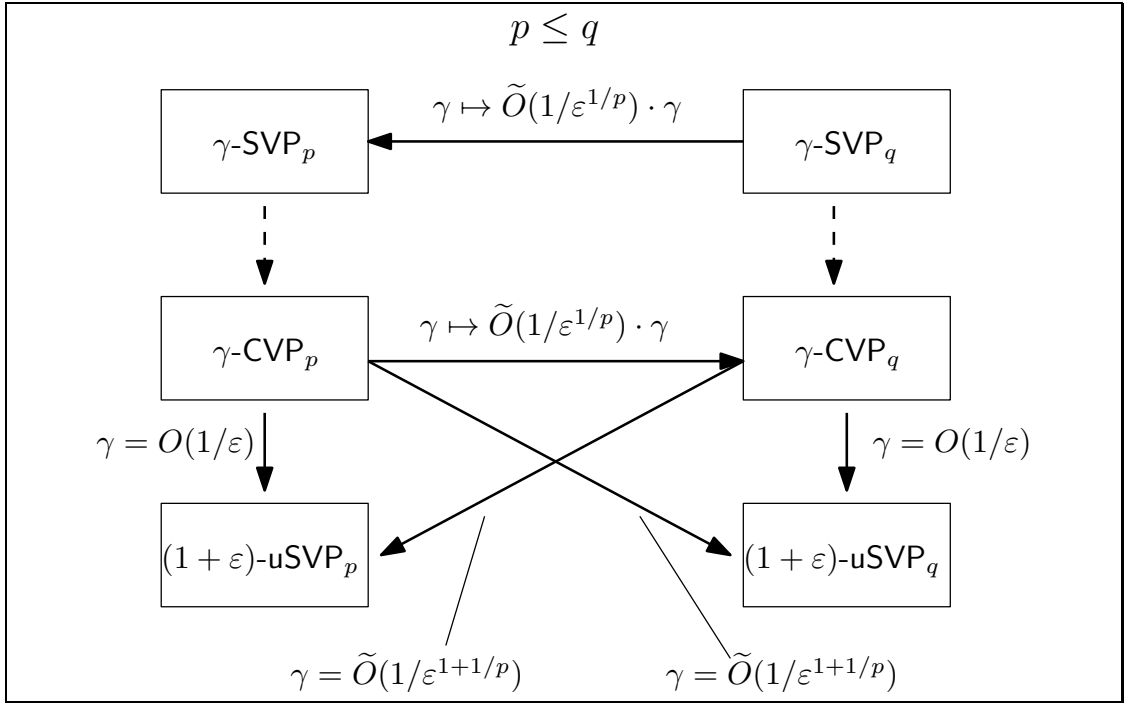}
\end{center}
\end{boxedminipage}
\end{center}
\caption{Dimension- and rank-preserving reductions between $\SVP$ and $\CVP$ in different norms. The dotted lines represent a polynomial-time reduction due to~\cite{GMSS99}, while solid lines represent $2^{\eps m}$-time reductions from this work. All reductions work for all $1 \leq p \leq q \leq \infty$.\label{fig:results}}
\end{figure}
\paragraph{SVP to SVP.} Our first main result shows that $O(1)\text{-}\SVP_p$ does in fact reduce to $O(1)\text{-}\SVP_2$ for all $p \geq 2$ (including $p = \infty$), in time $2^{\eps m}$. The running time of $2^{\eps m}$ of course makes this a non-standard reduction. But, in the most interesting settings, the fastest known algorithms run in time $2^{\Omega(n)} = 2^{\Omega(m)}$ (and we have some complexity-theoretic evidence suggesting that $2^{o(m)}$-time algorithms are impossible, as well as many cryptographic constructions that rely on the assumption that our algorithms cannot be improved by much), so that this type of non-standard reduction is ``almost as good as a polynomial-time reduction.''

In fact, our reduction is significantly more general. We reduce $O(\gamma)\text{-}\SVP_q$ to $\gamma\text{-}\SVP_p$ for any $1 \leq p \leq q \leq \infty$ in time $2^{\eps m}$, and more generally still, we show a smooth tradeoff between the running time of the reduction and the approximation factor. This shows in a very strong sense that $\SVP_q$ is no harder than $\SVP_p$ (up to a constant in the approximation factor).

\begin{theorem}[Informal, see Theorem~\ref{thm:svp_to_svp}]
\label{thm:svp_to_svp_intro}
For $p \leq q$, there is a $2^{\eps m}$-time dimension- and rank-preserving reduction from $\gamma'\text{-}\SVP_q$ to $\gamma\text{-}\SVP_p$, where 
\[
    \gamma' = \widetilde{O}(1/\eps^{1/p}) \cdot \gamma
    \; .
\]
\end{theorem}

To prove this result, we combine (slight generalizations of) the techniques from~\cite{EV20} together with a technique called lattice sparsification (originally due to Khot~\cite{Khot05}), which is a method of sampling a randomly chosen ``sparsified'' sublattice of a lattice. 

We first describe the relevant ideas from Eisenbrand and Venzin.~\cite{EV20} observe that the $\SVP_2$ subprocedure that they use actually outputs a list of $2^{\Omega(m)}$ lattice vectors that all lie inside the smallest $\ell_2$ ball that contains a shortest non-zero vector in the $\ell_q$ norm. Furthermore, this list is sampled from a distribution satisfying a certain technical non-degeneracy condition. For our purposes, it suffices to think of this non-degeneracy condition as follows: either (1) a shortest vector in the $\ell_q$ norm is likely to be in their list of vectors (in which case we are done); or (2) the list is likely to contain many distinct lattice vectors inside this $\ell_2$ ball. For $q \geq 2$, Eisenbrand and Venzin then use a covering argument to show that any list of $2^{\eps m}$ distinct vectors in this $\ell_2$ ball must contain a pair of distinct vectors that lie in an appropriately small $\ell_q$ ball. (We note that this covering argument is reminiscent of the elegant $M$-ellipsoid technique in~\cite{DPVEnumerativeLattice11}.) This immediately yields an algorithm, by checking all pairwise differences of the vectors in the list and outputting the result with the smallest (non-zero) $\ell_q$ norm.

The main technical contribution behind our result can be seen as a method for converting \emph{any} $\gamma\text{-}\SVP_p$ oracle into an oracle that samples a list of lattice vectors with appropriately short $\ell_p$ norms with a distribution satisfying a similar non-degeneracy condition. Combining this with a simple extension of the~\cite{EV20} covering argument to all $\ell_p$ norms immediately yields the above result. 

As an important special case, notice that vectors sampled independently from the uniform distribution over lattice vectors inside an appropriate $\ell_p$ ball certainly satisfy this condition.\footnote{To see this, notice that if the number of vectors in the ball is smaller than $N$, then $N$ samples will likely be enough to find any fixed vector in the ball---including the shortest vector in the $\ell_q$ norm. But, if there are many vectors in this ball, then $N$ independent samples from this distribution are likely to contain many distinct vectors.} \cite{Ste16} showed that the shortest vector in an appropriately sparsified lattice is more-or-less a uniformly random (primitive) lattice vector in a ball, so that an \emph{exact} $\SVP_p$ oracle can be used to sample a nearly uniformly random (primitive) lattice vector from an $\ell_p$ ball of any radius---by simply calling the oracle on an appropriately sparsified sublattice. So, it is not \emph{too} hard to use the ideas from~\cite{EV20} and~\cite{Ste16} to show that $O_\eps(1)\text{-}\SVP_q$ reduces to \emph{exact} $\SVP_p$ in $2^{\eps m}$ time.

To make this idea work with an \emph{approximate} $\SVP_p$ oracle is rather delicate because, unlike an exact $\SVP_p$ oracle, an approximate oracle will typically have a choice between many different vectors for its output.  In spite of this, we show that the lattice vector returned by the approximate $\SVP_p$ over a randomly sparsified sublattice satisfies essentially the same non-degeneracy conditions. Specifically, we show (using ideas from~\cite{Ste16}) that the only vectors that the oracle can choose with ``unexpectedly high'' probability are vectors that are integer multiples of some shortest vector $\vec{v}$, which are as good as the shortest vector from our perspective. On the other hand, if there does not exist a vector that is output with ``unexpectedly high'' probability, then the list of vectors sampled contains many distinct lattice vectors inside the $\ell_p$ ball.

\paragraph{CVP to SVP} Our second main result is a reduction from constant-factor approximate $\CVP_p$ to $\SVP_q$ for any $p, q \in [1,\infty]$ in time $2^{\eps m}$. In fact, this reduction works with a $(1+\eps)$-approximate $\SVP_q$ oracle as well, and even with a $(1+\eps)$-unique $\SVP_p$ oracle. Unique $\SVP_p$ is a potentially easier variant of $\SVP_p$, in which we are promised that there is only one solution (up to sign) to the $(1+\eps)\text{-}\SVP_p$ instance. (For example unique $\SVP_p$ is no harder than $\mathsf{GapSVP}_p$, the decision variant of $\SVP_p$, under an efficient rank-, dimension-, and approximation-factor-preserving reduction~\cite{LM09}. And, $(1+\eps)$-unique SVP is not even known to be NP-hard for constant $\eps$.)

\begin{theorem}[Informal, see Corollaries~\ref{cor:cvpq_to_usvp_p} and~\ref{cor:cvp_p_to_uSVP_q}]
\label{thm:cvp_to_svp_intro}
For any $p,q$, there is a $2^{O(\eps m)}$-time reduction from $\gamma$-$\CVP_q$ to $(1+\eps)$-unique $\SVP_p$, where
\[
    \gamma = \widetilde{O}(1/\eps^{1+1/\min\{p,q\}})
    \;.
\]
The reduction only calls its $\SVP$ oracle on lattices with rank $n+1$ and dimension $m+1$, where $n$ and $m$ are the rank and dimension of the input lattice respectively. 
\end{theorem}

Again, our main technical tools are lattice sparsification and (different) ideas introduced by Eisenbrand and Venzin~\cite{EV20}. In more detail, many $\CVP_p$ algorithms start with the observation that one can use an $\SVP_p$ oracle to find a lattice vector that is close to $k \vec{t}$ for some small integer $k$ by calling the $\SVP_p$ oracle on the lattice ``with $\vec{t}$ embedded in it,'' i.e., the lattice generated by
\begin{equation}
\label{eq:embedding}
    \begin{pmatrix}
    \basis &-\vec{t}\\
    0 &s
    \end{pmatrix}
    \; .
\end{equation}
Notice that short lattice vectors of the form $(\vec{v} - k \vec{t}, ks)$ in this new lattice correspond to lattice vectors close to $k\vec{t}$ in the lattice generated by $\basis$. (This idea is originally due to Kannan~\cite{Kannan87} and is often referred to as Kannan's  embedding.) 

A lot of work on $\CVP_p$ algorithms therefore naturally focuses on ways to force the algorithm to yield a solution with $k =1$. (E.g., this is how~\cite{AKS02} showed how to reduce $O(1)$-$\CVP_2$ to uniformly sampling short lattice vectors in time $2^{O(n)}$, and it is what we do in Section~\ref{sec:supergaussian}.) 

Eisenbrand and Venzin observed that it suffices to simultaneously find a lattice vector close to $k\vec{t}$ and a lattice vector close to $(k-1)\vec{t}$. By triangle inequality, the difference of these two vectors will be close to $\vec{t}$.\footnote{As far as the authors know,~\cite{EV20} was the first work to use the very natural idea, even in the case when $p = q$.} In fact, when one is reducing $\CVP_q$ to $\SVP_p$ for $p \neq q$, this method even seems preferable, since the fact that this technique outputs the difference of vectors that are close in the $\ell_p$ norm allows~\cite{EV20} to apply the ideas that we described above (i.e., the covering argument and the idea of  non-degenerate distributions) to this setting.

With this in mind, the high-level idea behind our reduction is quite simple: we use Kannan's embedding as in Eq.~\eqref{eq:embedding} (with a carefully chosen $s$), sample many random short lattice vectors from the resulting lattice using our $\SVP_p$ oracle and sparsification, and look for pairs of vectors whose difference is of the form $(\vec{y} - \vec{t}, s)$ with small $\ell_q$ norm. 

However, many subtleties arise here, even when reducing to exact $\SVP_p$. For example, if $\vec{x} $ is in the sparsified lattice, so are $-\vec{x}, \pm 2\vec{x}, \pm 3\vec{x}, \ldots, $. So, the event that $\vec{x}$ is in the sparsified lattice is of course not independent of the event that $2\vec{x}$ is in the sparsified lattice. These correlations are not an issue when trying to solve $\SVP$, since it never makes sense to output $k\vec{x}$ for $|k| \neq 1$ as a solution to $\SVP$ anyway. But, to make the above technique work, we might actually prefer, e.g., $2\vec{x}$ to $\vec{x}$, and we therefore must account for these correlations. We do so by using yet another sparsification-based technique---this one from~\cite{DRS14}---which allows us to limit the number of integer multiples of $\vec{x} \in \lat$ that can lie in an appropriately sized $\ell_p$ ball.

A more difficult issue is that it no longer suffices (as it did in the case  of $\SVP_p$) to show that there is some small $\ell_q$ ball that contains many vectors. Instead, we must show that there are many vectors close to $k\vec{t}$ (in the $\ell_p$ norm) and many vectors close to $(k-1)\vec{t}$, whose pairwise differences are close to $\vec{t}$ in the $\ell_q$ norm. This requires us to relate the number of vectors in two different groups, those close to $k \vec{t}$ and those close $(k-1)\vec{t}$ (and to worry about the $\ell_q$ distance to $\vec{t}$ of their difference). To do so, we observe that, simply by triangle inequality, each lattice vector at distance $r$ from $k\vec{t}$ naturally corresponds to a lattice vector at distance $r' = r + \dist(\vec{t},\lat)$ from $(k-1)\vec{t}$. We then very carefully choose parameters to argue that we can find a suitable not-too-large radius $r$ and choice of $k$ for which the number of vectors at distance $r$ from $(k-1)\vec{t}$ is not much smaller than the number of vectors at distance $r'$ from $(k-1)\vec{t}$. More specifically, we must choose our radius $r$ so that the total number of vectors in the embedded lattice generated by Eq.~\eqref{eq:embedding} that have $\ell_p$ norm at most $r$ is at most $2^{\eps m}$ times the number of vectors with norm $r'$. A packing argument shows that such a radius $r$ at which the number of lattice points ``grows slowly'' in this way must exist with $r \lesssim \dist(\vec{t},\lat)/\eps$.

Finally, in order to make our reduction work with an approximate (unique) $\SVP_p$ oracle, we show (using ideas from~\cite{SteSearchtodecisionReductions16}) that the solution to the sparsified $(1+\eps)\text{-}\SVP_p$ instance will be unique with probability roughly $2^{-\eps m}$ if (and only if) the number of lattice points ``grows slowly'' at $r$. So, conveniently, this slow-growing property proves to be exactly what we need to resolve two different issues. 

Intuitively, we cannot hope to work with a $\gamma\text{-}\SVP_p$ oracle for larger values of $\gamma$ because, e.g., our oracle might only output vectors of the form $(\vec{v} - 2k\vec{t}, 2ks)$. In contrast, Eisenbrand and Venzin are able to use a specific $\gamma\text{-}\SVP_2$ \emph{algorithm} (rather than a generic oracle) for large constant $\gamma$ by taking advantage of specific properties of the algorithm that prevent this from happening.

\paragraph{CVP to SVP for $p = q$. } Our reduction from $\CVP_p$ to $\SVP_q$ is even interesting in the case when $p = q$ and when the $\SVP$ oracle is exact. Indeed, many works (such as~\cite{AKS02,BNSamplingMethods07,ADRS15}) have shown how to adapt specific $2^{Cm}$-time algorithms for $\SVP_p$ to work for constant-factor-approximate $\CVP_p$ with the same or nearly the same running time. Our result shows that this can be done generically, up to a factor of $2^{\eps m}$ in the running time. Given this importance, we study this special case separately and prove the following theorem, which gives quantitatively stronger results than Theorem~\ref{thm:cvp_to_svp_intro} for $1 \leq p \leq 2$ when the $\SVP_p$ oracle is exact. In particular, plugging in $p = 2$ and $\eps = \log m/m$ strictly improves on Kannan's celebrated polynomial-time reduction from $\sqrt{m}\text{-}\CVP_2$ to $\SVP_2$.

\begin{theorem}[Informal, see Theorem~\ref{thm:CVP_to_SVP_supergaussian}]
    For $1 \leq p \leq 2$, there is a $2^{\eps m}$-time reduction from $\gamma\text{-}\CVP_p$ to (exact) $\SVP_p$, where 
\[
    \gamma = O(1/\eps^{1/p})
    \; .
\]
The reduction calls its $\SVP_p$ oracle on lattices with rank $n+1$ and dimension $m+1$, where $n$ and $m$ are the rank and dimension of the input lattice respectively.
\end{theorem}

Behind this theorem are generalizations of discrete-Gaussian-based techniques from~\cite{ADRS15} and sparsification techniques from~\cite{Ste16}. Specifically, we show how to use sparsification and an (exact) $\SVP_p$ oracle to sample from an $\ell_p$ analogue of a Gaussian distribution over a lattice, which was shown in~\cite{Ste16} for the case $p = 2$ (and sketched for the more general case). We then show that $2^{\eps m}$ such samples from the embedded lattice discussed above (with the width of the distribution chosen appropriately) suffices to solve $O(1/\eps^{1/p})\text{-}\CVP_p$, which was shown for the case $p = 2$ and a specific choice of $\eps$ in~\cite{ADRS15}. This can also be seen as a variant of the original~\cite{AKS02} $2^{O(n)}$-time reduction from $O(1)\text{-}\CVP_2$ to the problem of sampling nearly uniformly from lattice points with bounded norm (and Bl{\"o}mer and Naewe's generalization to $\ell_p$ norms~\cite{BNSamplingMethods07}), in which we use sparsification to do the random sampling and choose our radius very carefully (using the oracle).

The result only holds for $1 \leq p \leq 2$ for a rather technical reason: the $\ell_p$ analogue $e^{-\|\vec{x}\|_p^p}$ of the Gaussian is a positive definite function if and only if $1 \leq p \leq 2$. It is unclear whether this issue is inherent.

\paragraph{CVP to CVP} Our final main result is a reduction between $\CVP$ in different norms. Specifically, we reduce $O_\eps(\gamma)\text{-}\CVP_p$ to $\gamma\text{-}\CVP_q$ in $2^{\eps m}$ time, for $p \leq q$. (Notice that our $\SVP$ reduction went from $q$ to $p$, i.e., from big to small, while our $\CVP$ reduction goes from $p$ to $q$, or from small to big.)

\begin{theorem}[Informal, see Theorem~\ref{thm:cvp_q<p}]
\label{thm:CVP_to_CVP_intro}
 For $1 \leq p \leq q \leq \infty$, there is a dimension- and rank-preserving $2^{\eps m}$-time reduction from $\gamma'\text{-}\CVP_p$ to $\gamma\text{-}\CVP_p$, where 
\[
    \gamma' = \widetilde{O}(1/\eps^{1/p}) \cdot \gamma
    \; .
\]
\end{theorem}

The reduction behind Theorem~\ref{thm:CVP_to_CVP_intro} uses a very simple and natural idea that, to the authors' knowledge, was also first published in~\cite{EV20}. The idea is to randomly perturb the input target point $\vec{t} \in \R^m$ to $\vec{t}'$ and then to call our $\CVP_q$ oracle on the perturbed point. One needs to choose the method of perturbation so that with probability $2^{-\eps m}$, the perturbed target will be very close in $\ell_q$ norm to a closest lattice vector to $\vec{t}$ in the $\ell_p$ norm. Though they do not describe it this way, the algorithm in~\cite{EV20} for $p < 2$ can already be viewed as such a reduction for $q = 2$. Therefore, our contribution here is largely the extension to arbitrary $q$, as well as a simpler reduction and analysis.

\section{Preliminaries}

\subsection{Notation}
We use $\real,\intg,\nat$ to denote the real numbers, integers and  natural numbers respectively. For any natural number $m$, we use $[m]$ to denote the set $\{1,2,\hdots, m\}$. For any $p\in [1,\infty)$, the $\ell_p$ norm on $\real^m$ is defined as 
\[ \|\vect{x}\|_p =\left(\sum\limits_{i=1}^m \abs{x_i}^p \right)^{\frac{1}{p}}. \]
The $\ell_\infty$ norm is defined as $\|\vect{x}\|_\infty =\mathrm{max}_{i\in [m]} \{x_i\} $. For any $1\leq p\leq \infty $, we use $\mathcal{B}_p^m$ to denote the closed unit ball in $\ell_p$ norm in $\real^m$ i.e.
\[ \mathcal{B}_p^m=\{\vect{x}\in \real^m :\; \|\vect{x}\|_p \leq 1 \}.\]

For a matrix $\basis \in \R^{m \times n}$ with rank $n$, we abuse notation and write $\basis^{-1} := (\basis^T \basis)^{-1} \basis^T$ for the left inverse of $\basis$. In particular, $\basis^{-1} \basis \vec{x} = \vec{x}$ for all $\vec{x} \in \R^n$.

\subsection{Lattices}
For any set of $n$ linearly independent vectors $\mathbf{B}=\{\vect{b_1},\hdots,\vect{b_n}\}$ from $\real^m$, the lattice $\cL$ generated by basis $\mathbf{B}$ is
\[ \cL(\mathbf{B})= \left\{ \sum\limits_{i=1}^n z_i\vect{b_i} : z_i \in \intg  \right\}.\]

We call $n$ the \textit{rank} of the lattice $\cL$ and $m$ the \textit{dimension}. The vectors $\mathbf{B}=\{\vect{b}_1,\ldots,\vect{b}_n\}$ forms a \textit{basis} of the lattice. Given a basis $\basis$, we use $\cL(\basis)$ to denote the lattice generated by $\basis$.
If $n=m$, we say the lattice $\cL$ is \textit{full rank}. We define the length of the shortest  non-zero vector under $\ell_p$ norm by 
\[\lambda_1^{(p)}(\cL)= \underset{{\vect{x\in \cL\setminus \{\vect{0}\}}}}{\mathrm{min}} \|\vect{x}\|_p .\]
For any $\vect{t}\in \real^n$, we define the distance to the closet lattice vector under $\ell_p$ norm by \[\dist_p(\vect{t,\cL})=\underset{{\vect{x}\in\cL}}{\mathrm{min}} 
\; \|\vect{x-t}\|_p.\]
We also write
\[
    \lambda_2^{(p)}(\cL) := \min \{r \ : \ \dim(\spn(\lat \cap \, r\cB_p^m)) \geq 2 \}
    \; .
\]
(I.e., $\lambda_2^{(p)}(\lat)$ is the minimal length of a vector that is linearly independent from at least one shortest vector.)

A vector $\vect{v} \in \cL$ is a \textit{non-primitive} lattice vector if there exists an $\vect{x}\in \cL$ and a $z > 1$ such that $\vect{v}=z\vect{x}$. Otherwise it is \textit{primitive}. For any lattice $\cL$, we use $\cL^{prim}$ to denote the set of all the primitive lattice vectors in $\cL$. 

\begin{claim}\label{clm:bound_on_set_S'}
For any lattice $\cL\subset \real^m$, radius $r>0$, and set $\cS \subseteq \cL_{\neq \vec{0}} \cap \, r\mathcal{B}_p^m$, then 
\[ |\cS'|\geq \frac{\lambda_1^{(p)}(\cL)}{r}|\cS|\]
where $\cS':=\{\vec{v}\in \cL^{prim} : \exists k\in \intg_{> 0}, k\vec{v}\in \cS\}$.
\end{claim}
\begin{proof}
 For any primitive vector  $\vec{v} \in \cL^{prim}$, the number of vectors in $\cL \cap \, r\mathcal{B}_p^m$ that are a positive integer multiple of $\vec{v}$ is at most $\frac{r}{\|\vec{v}\|_p} \leq \frac{r}{\lambda_1^{(p)}(\cL)}$. The claim follows. 
\end{proof}

For any $p \ge 1$, set of vectors $\mathcal{A}$, target vector $\vect{t}$ and radius $r>0$, we use $N_p(\mathcal{A},r,\vect{t})$ to denote the number of vector in set $\mathcal{A}$ whose $\ell_p$ norm is at most $r$.
\[N_p(\mathcal{A},r,\vect{t})= |\{ \vect{x}\in \mathcal{A} : \|\vect{x-t}\|_p\leq r \} | 
\; .\]
We omit the parameter $\vect{t}$, when $\vect{t}=\vect{0}$.
\begin{lemma}[\cite{Ste16}, Corollary 2.3]\label{lem:bound_on_number_of_vectors}
For any lattice $\cL\subset \ratn^m$, with basis $(\vec{b}_1,\hdots,\vec{b}_n)$, $\vec{t}\in \ratn^{m}$ and radius $r>0$, let $\ell$ be the bound on the
bit length of the $\vec{b}_i$
for all $i$ in the natural representation of rational numbers. Then,
\[|(\cL-\vec{t})\cap \, r\mathcal{B}_p^m | \leq 1+(2+r)^{\poly(m,\ell)}\; .\]
\end{lemma}

We will also need the following simple claim.

\begin{lemma}
\label{lem:lambda1_projection}
    For any lattice $\lat \subset \R^n$ (for $n \geq 2$) and $\vec{x} \in \lat_{\neq \vec0}$, let $\pi(\lat)$ be the lattice obtained by projecting $\lat$ orthogonal to $\vec{x}$. Then, 
    \[
        \lambda_1^{(2)}(\pi(\lat)) \geq \frac{3}{4} \cdot \frac{\lambda_1^{(2)}(\lat)^2}{\|\vec{x}\|_2}
    \]
\end{lemma}
\begin{proof}
    Let $\vec{y}\in \lat$ such that $\pi(\vec{y}) \neq \vec0$, and let $\lat(\vec{x},\vec{y})$ be the lattice generated by $\vec{x}$ and $\vec{y}$, which has determinant $\det(\lat(\vec{x}, \vec{y})) = \|\pi(\vec{y})\| \|\vec{x}\|_2$. It is, however, well known that any two-dimensional lattice $\lat'$ satisfies $\lambda_1^{(2)}(\lat')^2 \leq 4/3 \cdot \det(\lat')$.
    Therefore,
    \[
        \lambda_1^{(2)}(\lat)^2 \leq \lambda_1^{(2)}(\lat(\vec{x}, \vec{y}))^2 \leq \frac{4}{3} \cdot \|\vec{x}\|_2 \|\pi(\vec{y})\|_2
        \; .
    \]
    Rearranging shows that $\|\pi(\vec{y})\|_2 \geq 3\lambda_1^{(2)}(\lat)^2/(4 \|\vec{x}\|_2)$, as needed.
\end{proof}

\subsection{The discrete supergaussian distributions}
For any $\ell_p$ norm, we define the function $f_p:\real^m \mapsto \real$ as $f_p(\vect{x})=\exp(-\|\vect{x}\|_p^p)$, which is also known as a \emph{supergaussian}. For a discrete set $\mathcal{A}\subset \real^m$, we define $f_p(\mathcal{A})=\sum\limits_{\vect{x}\in A} f_p(\vect{x})$. 
\begin{defn}
For a $n$-rank lattice $\cL\subset \real^m$, we define $D_{\cL,p}$  as the probability distribution over $\cL$ such that probability of  drawing $\vect{x}\in \cL$ is proportional to $f_p(\vect{x})$:
\[\underset{X\sim D_{\cL,p}}{Pr}[X=\vect{x}]=\frac{f_p(\vect{\vect{x}})}{f_p(\cL)} .\]
\end{defn}

We will need the following tail bound on $D_{\cL,p}$ for $1 \leq p \leq 2$, proven by~\cite{MSKissing18}. It is a generalization of Banaszczyk's celebrated result for $p = 2$. 

\begin{lemma}[Lemma 3.11,\cite{MSKissing18}\label{lem:suppergaussian}]
	For any lattice $\cL \subset \real^m$, $0<p\leq 2$ and $a\geq 1$,
	\[\sum\limits_{\vect{x}\in \cL,\, \|\vect{x}\|_p\geq a(m/p)^{1/p}} f_p(\vect{x})\leq (ea^pe^{-a^p})^{m/p} \cdot f_p(\lat) \; .\]
\end{lemma}

We will also need the following lemma. The special case of $p = 2$ is again a celebrated inequality due to Banaszczyk~\cite{Ban93}.

\begin{lemma}\label{lem:shifted_supergaussian_mass}
    For any $1 \leq p \leq 2$, any lattice $\lat \subset \R^m$, and any $\vec{t} \in \R^m$,
    \[
        f_p(\vec{t}) f_p(\lat) \leq f_p(\lat + \vec{t}) \leq f_p(\lat) \; .
    \]
\end{lemma}
\begin{proof}
The upper bound follows from the fact that $f_p$ is a positive-definite function. (See, e.g.,~\cite{MSKissing18}.)
    For the lower bound, we have
    \[
        f_p(\lat + \vec{t}) = f_p(\lat + \vec{t})/2 + f_p(\lat - \vec{t})/2 =  \sum_{\vec{v} \in \lat} (e^{-\|\vec{v} + \vec{t}\|_p^p}/2 + e^{-\|\vec{v}-\vec{t}\|_p^p}/2)
        \; .
    \]
    Then,
    \[
        e^{-\|\vec{v} + \vec{t}\|_p^p}/2 + e^{-\|\vec{v}-\vec{t}\|_p^p}/2 = e^{-\|\vec{v} + \vec{t}\|_p^p/2 - \|\vec{v}- \vec{t}\|_p^p/2}\cosh(\|\vec{v}+\vec{t}\|_p^p/2 - \|\vec{v}-\vec{t}\|_p^p/2) \geq e^{-\|\vec{v} + \vec{t}\|_p^p/2 - \|\vec{v}- \vec{t}\|_p^p/2}
        \; .
    \]
    It therefore suffices to prove that $|v + t|^p/2 + |v-t|^p/2 \leq |v|^p + |t|^p$
    for all $v, t \in \R$. 
    
    We may assume without loss of generality that $0 \leq t \leq v$. 
    Then, the necessary inequality follows from the fact that for such $v$ and $t$,
    \begin{align*}
        (v+t)^p/2 + (v-t)^p/2 
        -v^p 
            &= t^p \cdot \sum_{i=1}^\infty \frac{t^{2i-p}}{v^{2i-p}} \cdot \frac{(p)_{2i}}{(2i)!}  \\
            &\leq t^p
        \; .
    \end{align*}
    Here, the equality is the Taylor series around $t=0$ (which converges for $|t| \leq v$), using the notation $(a)_i := a(a-1)\cdots (a-i+1)$. The inequality follows from direct computation, or it can be derived from the fact that the worst case is clearly the case when $t=v$, in which case the result is trivial.
\end{proof}

\subsection{Lattice problems}
\begin{defn}
For any $\gamma=\gamma(n,m)\geq 1$ and $1\leq p \leq \infty$, the $\gamma$-approximate Shortest Vector Problem ($\gamma$-$\svp_p$) is the search problem defined as: The input is a basis $\basis \in \real^{m\times n} $ of the $n$-rank lattice $\cL$. The goal is to output a vector $\vect{v}\in \cL$ such that $0<\|\vect{v}\|_p \leq \gamma \lambda_1^{(p)}(\cL)$.
\end{defn}

\begin{defn}
For any $\gamma = \gamma(n,m) \geq 1$ and $1\leq p \leq \infty$, the $\gamma$-unique Shortest Vector Problem ($\gamma$-$\uSVP_p$) is the promise search problem defined as follows. The input is a basis $\basis \in \real^{m\times n}$ of the $n$-rank lattice $\cL$ with the promise that $\gamma \lambda_1^{(p)}(\lat) < \lambda_2^{(p)}(\lat)$. The goal is to output $\vec{v} \in \cL$ such that $\|\vec{v}\|_p = \lambda_1^{(p)}(\cL)$.
\end{defn}

\begin{defn}
For any $\gamma=\gamma(n,m)\geq 1$ and $1\leq p \leq \infty$, the $\gamma$-approximate Closest Vector Problem $\gamma$-$\cvp_p$ is the search problem defined as: The input is a basis $\basis \in \real^{m\times n}$ of the lattice $\cL$ and a target vector $\vect{t}$. The goal is to output a vector $\vect{v}\in \cL$ such that $\|\vect{v}-\vect{t}\|_p \leq \gamma \cdot \dist_p(\vect{t},\cL)$.
\end{defn}

\begin{defn}
    For $\alpha = \alpha(n,m) > 0$, $\gamma = \gamma(n,m) \geq 1$, and $p \geq 1$, the $(\alpha,\gamma)$-Bounded Distance Decoding problem ($(\alpha, \gamma)\text{-}\BDD_p$) is the promise search problem defined as follows. The input is a lattice $\lat \subset \R^m$ of rank $n$ and a target vector $\vec{t} \in \R^m$ with $\dist_p(\vec{t}, \lat) < \alpha \cdot \lambda_1^{(p)}(\lat)$. The goal is to output $\vec{v} \in \lat$ with $\|\vec{v} - \vec{t}\|_p \leq \gamma \cdot \dist_p(\vec{t},\lat)$.
\end{defn}

When $\gamma=1$, we simply write $\svp_p$, $\cvp_p$, and $\alpha$-$\BDD_p$ respectively.

\begin{defn}
For any $\delta=\delta(n,m)>0$, $1\leq p\leq \infty$ and $M=M(n)\in \nat$, $\delta$-$\dss_p^M$ (Discrete Supergaussian Sampling) is defined as follows: The input is a basis $\basis \in \real^{m\times n}$ for a lattice $\cL$ of rank $n$. The goal is to output  $M$ i.i.d samples from a distribution $\tilde{D}$ with the property that for any lattice vector $\vec{y}\in \cL \cap \, m^{1/p} \cB_p^m $,
\[ \Pr\limits_{X \sim \tilde{D}}[\vec{X}=\vec{y}] \geq e^{-\delta}\cdot\Pr\limits_{\vec{X}'\sim D_{\cL,p}}[\vec{X}' = \vec{y}]
\; .
\]
\end{defn}

\begin{defn}
For any parameter $\beta\geq 0$,$\gamma\geq 1$ and $1\leq p\leq \infty$, $(\beta,\gamma)$-$\gappvcp_p$ (the Primitive Vector Counting Problem) is the promise problem defined as follows: the input is a basis $\basis \in \real^{m\times n}$ for a lattice $\cL$, radius $r>0$ and an integer $N\geq 1$. It is a NO instance if $N_p(\cL^{prim},r)\leq N$ or if $\lambda_1^{(p)}(\cL)\leq \frac{\beta r}{N}$ and a YES instance if $N_p(\cL^{prim},r)> \gamma N$.
\end{defn}

\subsection{Sparsification}

Here, we present some results based on ideas from~\cite{Ste16,SteSearchtodecisionReductions16}.

\begin{restatable}[{\cite[Lemma 2.16]{Ste16}}]{lemma}{almostindependent}
\label{lem:almostindependent}
For any prime $Q$ and collection of vectors $\vec{x}, \vec{v}_1,\ldots, \vec{v}_N \in \Z_Q^n \setminus \{\vec0 \}$ such that $\vec{x}$ is not a scalar multiple of any of the $\vec{v}_i$, we have
\[
\frac{1}{Q} - \frac{N}{Q^2} \leq \Pr\big[\langle \vec{z}, \vec{x}\rangle = 0 \bmod Q \text{ and } \langle \vec{z}, \vec{v}_i \rangle \neq 0 \bmod Q \ \forall i \big] \leq \frac{1}{Q}
\; ,
\]
where $\vec{z} $ is sampled uniformly at random from $\Z_Q^n$.
\end{restatable}

The following result is implicit in~\cite[Theorem 3.3]{SteSearchtodecisionReductions16}. We include a proof in Appendix~\ref{proof:uniformSamplingByuSVP} for completeness.

\begin{restatable}{theorem}{uniformSampllingByuSVP}\label{thm:uniform_samplling_by_uSVP}
    For any $\gamma = \gamma(n,m) \geq 1$ and efficiently computable function $f(m) \geq 10$, there is a polynomial-time algorithm with access to a $\gamma\text{-}\uSVP_p$ oracle that takes as input $p \geq 1$, a (basis for a) lattice $\cL \subset \R^m$ of rank $n$, a radius $\lambda_1^{(p)}(\cL) \leq r < f(m) N_p(\cL^{prim},r) \lambda_1^{(p)}(\lat)$, and an integer $N \geq 10$, and outputs a vector $\vec{y} \in \cL$ such that, if $N_p(\cL^{prim},\gamma r)/f(m) \leq N \leq N_p(\cL^{prim}, r)$, then for any $\vec{x} \in \cL^{prim} \cap \, \, \mathcal{B}_p(r)$,
    \[
        \Pr[\vec{y} = \vec{x}] \geq \frac{1}{1000f(m) N \log(f(m)N)}\; . %
    \]
    
    The algorithm only calls its oracle on sublattices of the input lattice. In particular, it preserves the rank and dimenson of the lattice.
\end{restatable}

\subsection{A reduction from CVP to BDD (above the unique decoding radius)}

We now present an $\ell_p$ generalization of~\cite[Theorem 6.1]{DRS14}, whose proof we defer the to Appendix~\ref{proof:CVPtoCVPP}.

\begin{restatable}{theorem}{CVPtoCVPP}\label{theorem:CVPtoCVPP}
    For any $p \in [1,\infty]$, (efficiently computable) $\tau = \tau(m) > 0$, and $\gamma = \gamma(m) \geq 1$, there is an efficient reduction from $(1+1/\tau)\gamma \text{-}\CVP_p$ to $(\alpha, \gamma)\text{-}\BDD_p$, where $\alpha := 1+\tau$.
\end{restatable}

\section{A reduction from \texorpdfstring{$\gamma'$-$\svp_q$}{SVP} to \texorpdfstring{$\gamma$-$\svp_p$}{SVP}}

We first show our reduction from $\SVP_p$ to $\SVP_q$ for $1 \leq p \leq q \leq \infty$. We will need the following covering lemma. It is a slight generalization of the main geometric lemma in~\cite{EV20}.

\begin{lemma}{\label{lem:packing_all_norms}}
For any $q\geq p\geq 1$, and $\alpha \geq e$, $ m^{1/p-1/q}\mathcal{B}_p^m$ can be covered by $(e^4 \alpha^p)^{m/\alpha^p}$ translated copies of $\alpha \mathcal{B}_{q}^m$. 
\end{lemma}
\begin{proof}
    We will show that $(e^4\alpha^p)^{m/\alpha^p}$ translated copies of $m^{1/q}B_q^m$ are sufficient to cover $r B_p^m$ for $r := m^{1/p}/\alpha$, which is equivalent but a bit more convenient. 
    Let $S := \mathbb{Z}^m \cap \,  r B_p^m$. Notice that for every $\vec{t} \in  r B_p^m$, there exists $\vec{z} \in S$ such that $\|\vect{z} - \vect{t}\|_q \leq m^{1/q}$. (This point can be found explicitly by simply rounding all of the coordinates of $\vec{t}$ towards zero.) In other words, translates of $m^{1/q}B_q^m$ centered at the points in $S$ cover all of $rB_p^m$, and it therefore suffices to bound $|S|$.

Let
\[
    \Theta_p(\tau) := \sum_{z \in \mathbb{Z}} e^{-\tau |z|^p}
    \; .
\]
For any $\tau > 0$, we have
\[
 |S| \leq e^{\tau  r^p}	\sum_{\vect{z}\in \mathbb{Z}^m} e^{-\tau \|\vect{z}\|_p^p} = e^{\tau r^p} \Theta_p(\tau)^m
 \; ,
\]
where the equality follows from the fact that $f(\vect{z}) := e^{-\tau \|\vect{z}\|_p^p} = \prod_i e^{-\tau |z_i|^p}$ is a product measure (and the inequality follows from an averaging argument).

It remains to bound $\Theta_p(\tau)$. Indeed,
\[
	\Theta_p(\tau) \leq 1 + 2e^{-\tau} + 2 \int_{1}^\infty e^{-\tau |x|^p} {\rm d} x \leq 1 + 2e^{-\tau} + 2 \int_{1}^\infty x^{p-1}e^{-\tau |x|^p} {\rm d} x = 1 + 2e^{-\tau} \cdot( 1+ 1/(p \tau))
	\; .
\]
Plugging in $\tau := p\log \alpha$ gives
\[
    |S| \leq \alpha^{p m/\alpha^p} \cdot (1+ 2\alpha^{-p} (1+1/(p^2 \log \alpha))^m \leq (e^4\alpha^p)^{m/\alpha^p} 
    \; ,
\]
as needed,
where we have used the inequality,
\[
    (1+ 2(1+1/(p^2 \log \alpha) )/\alpha^p) \leq e^{2 (1+1/(p^2 \log \alpha))/\alpha^p} \leq e^{4/ \alpha^p}
    \; ,
\]
valid for $\alpha \geq e^{1/p^2}$.
\end{proof}

We can now present our reduction between SVP in different norms. The proof as presented below is a simplification (suggested by Moritz Venzin) of our original proof.

\begin{theorem}
\label{thm:svp_to_svp}
For any (efficiently computable) $\epsilon = \eps(n,m) \in (0,1/100)$, $q\geq p\geq 1$, and  $\gamma = \gamma(n,m) \geq 1$, there is a $(\gamma^4 \cdot 2^{\eps m}\poly(m))$-time reduction from $\gamma'\text{-}\SVP_q$ to $\gamma\text{-}\SVP_p$, where 
\[
    \gamma' := 100 \log^{1/p}(1/\eps) \cdot \gamma/\eps^{1/p}
    \; .
\]
The reduction preserves dimension and rank and only calls its oracle on sublattices of the input lattice.
\end{theorem}
\begin{proof}
We show a reduction that runs in polynomial time and finds a sufficiently short vector with probability at least $2^{-\eps m}/(C\gamma^4 m^8)$. The result follows by running this algorithm $2^{\eps m} \gamma^4 \cdot \poly(m)$ times.

The reduction takes as input a basis $\basis \in \R^{m \times n}$ for a lattice $\lat \subset \R^m$ and behaves as follows. It first finds a prime $Q$ with $10\gamma m^2 2^{\eps m/4} \leq Q \leq 20\gamma m^2 2^{\eps m/4}$ and samples $\vec{z}_1,\vec{z}_2 \in \Z_Q^n$ uniformly and independently at random. The reduction then sets 
\[
    \lat_i := \{ \vec{y} \in \lat \ : \ \langle \vec{z}_i, \basis^{-1} \vec{y} \rangle = 0 \bmod Q \}
    \; ,
\]
and calls its $\gamma$-$\SVP_p$ oracle on $\lat_1$ and $\lat_2$, receiving as output $\vec{v}_1, \vec{v}_2$. The reduction then outputs either $\vec{v}_1 - \vec{v}_2$ or $\vec{v}_1/k$ where $k \geq 1$ is maximal such that $\vec{v}_1/k \in \lat$---whichever is has smaller $\ell_q$ norm. (If $\vec{v}_1 = \vec{v}_2$, then the reduction outputs $\vec{v}_1/k$.)

The reduction clearly runs in polynomial time. It suffices to show that it succeeds with probability at least $1/Q^4$.

To prove correctness, let $\vec{x} \in \lat$ be a shortest non-vector in the $\ell_q$ norm, i.e., $\|\vec{x}\|_q = \lambda_1^{(q)}(\lat)$, and let $r_p := m^{1/p-1/q}\lambda_1^{(q)}(\lat) \geq \|\vec{x}\|_p$ and $r_q := \gamma' \lambda_1^{(q)}(\lat)$. Notice that $\lat_1$ and $\lat_2$ are independent and identically distributed random variables. It follows that $\vec{v}_1$ and $\vec{v}_2$ are also independent and identically distributed. By the correctness of the oracle, whenever $\vec{x} \in \lat_i$, we must have $\|\vec{v}_i\|_p \leq \gamma r_p$. We will therefore study the distribution of $\vec{v}_i$ conditioned on the event that $\vec{x} \in \lat_i$. It follows immediately from Lemma~\ref{lem:almostindependent} that $\Pr[\vec{x} \in \lat_i] = 1/Q$, so that we can afford to condition on this event. We then divide our analysis into two cases:  one for when this distribution has a high-probability vector and one when it does not.

First, suppose there exists a $\vec{w} \in \lat \cap \, \gamma r_p B_p^m$ such that
\[
\Pr[\vec{v}_i = \vec{w} \text{ and } \vec{x} \in \lat_i] > 1/Q^2
\; .
\]
By Lemma~\ref{lem:almostindependent}, we must have $\vec{w} = \alpha \vec{x} \bmod Q \lat$ for some $\alpha \in \Z_Q$, i.e., $\vec{w} = \alpha \vec{x} + Q \vec{y}$, where $\vec{y} \in \lat$. We claim that $\vec{w}$ must be a scalar multiple of $\vec{x}$ (not just modulo $Q$). If this is the case, then we see that with probability at least $1/Q^2$, $\vec{v}_1 = \vec{w}$ in which case $\vec{x} = \vec{v}_1/k$, and we are done. So, suppose not. Then,
\[
    \|\vec{w}\|_2 \geq Q\|\pi_{\vec{x}^\perp}(\vec{w})\|_2 = Q\|\pi_{\vec{x}^\perp}(\vec{y})\|_2 \geq Q\lambda_1^{(2)}(\pi_{\vec{x}^\perp}(\lat))
    \; .
\]
By Lemma~\ref{lem:lambda1_projection}, this implies that $\|\vec{w}\|_2 \geq (3/4) \cdot Q \lambda_1^{(2)}(\lat)^2/\|\vec{x}\|_2 >(Q/\sqrt{m}) \lambda_1^{(2)}(\lat)$, i.e. $\|\vec{w}\|_p > \gamma \lambda_1^{(p)}(\lat)$. This contradicts the correctness of the oracle. So, $\vec{w}$ must be a multiple of $\vec{x}$, as claimed.

Second, suppose that 
\[
    \Pr[\vec{v}_i = \vec{w} \text{ and } \vec{x} \in \lat_i] \leq 1/Q^2
\]
for all $\vec{w}$. 
In this case, it suffices to show that there exists some $\vec{c} \in \R^m$ such that
\[
    \Pr[\vec{v}_1, \vec{v}_2 \in S_{\vec{c}} \text{ and } \vec{v}_1 \neq \vec{v}_2] \geq 1/Q^4
    \; ,
\]
where $S_{\vec{c}} := ((r_q/2)\mathcal{B}_q^m + \vec{c})$.
Indeed, whenever this event occurs, $\vec{v}_1 - \vec{v}_2 \in \lat$ is a short non-zero vector.

By Lemma~\ref{lem:packing_all_norms}, $\gamma r_p \mathcal{B}_p^m$ can be covered by $(e^4  (\gamma'/(2\gamma))^p)^{m(2 \gamma/\gamma')^p} \leq Q/2$ translated copies of $(r_q/2) \cdot B_q^m$. Since whenever $\vec{x} \in \lat_i$, we have $\vec{v}_i \in \gamma r_p\mathcal{B}_p^m$, it follows from the pigeonhole principle that there must exist some $\vec{c}$ with
\[
    \Pr[\vec{v}_i \in S_{\vec{c}} \ | \ \vec{x} \in \lat_i] \geq 2/Q
    \; .
\]
Therefore,
\[
    \Pr[\vec{v}_i \in S_{\vec{c}} \text{ and } \vec{x} \in \lat_i] \geq 2/Q^2
    \; .
\]
Finally, we conclude that
\begin{align*}
    \Pr[\vec{v}_1, \vec{v}_2 \in S_{\vec{c}} \text{ and } \vec{v}_1 \neq \vec{v}_2] 
        &\geq \Pr[\vec{v}_1, \vec{v}_2 \in S_{\vec{c}} \text{ and } \vec{v}_1 \neq \vec{v}_2 \text{ and } \vec{x} \in \lat_1 \cap \, \lat_2] \\
        &\geq \Pr[\vec{v}_1 \in S_{\vec{c}} \text{ and } \vec{x} \in \lat_1] \\
        &\qquad \cdot \big( \Pr[\vec{v}_2 \in S_{\vec{c}}  \text{ and } \vec{x} \in \lat_2] - \max_{\vec{w}} \Pr[\vec{v}_2 = \vec{w} \text{ and } \vec{x} \in \lat_2] \big)\\
        &\geq (2/Q^2) \cdot (2/Q^2 - 1/Q^2)\\
        &\geq 1/Q^4
        \; ,
\end{align*}
as needed, where the second inequality uses the independence of $(\vec{v}_1, \lat_1)$ from $(\vec{v}_2, \lat_2)$.
\end{proof}

\section{Reductions from CVP to SVP (and to CVP) in different norms}
In this section, we present a reduction from $\gamma$-$\cvp_q$ to ($(1+\eps)$-unique) $\svp_p$ for any $p$ and $q$. We also present a reduction from $\gamma'$-$\cvp_p$ to $\gamma$-$\cvp_q$ where $1 \leq p \leq q \leq \infty$ and $\gamma' = O(\gamma)$. 

We will first need a basic lemma about the growth of the number of lattice points in a convex body as we increase the radius of the body. We include the simple packing-based proof for completeness.

\begin{lemma}\label{lem:counting}
	For any radii $0<r<R$, any lattice $\cL \subset \real^n$ and any symmetric convex body $\mathcal{K}\subset \real^n$
	
	\[ |\cL\cap \, R \mathcal{K}| \leq \left (1+\frac{2R}{r}\right )^n |\cL \cap \, r\mathcal{K}|.\]
	
\end{lemma}
\begin{proof}
	Notice that for any $\vect{x}\in \real^n$,
	\[ \left|\{ \vect{y}\in \cL : \vect{x}\in \vect{y}+\frac{r}{2}\mathcal{K} \} \right| =\left|\cL \cap \, \left( \frac{r}{2}\mathcal{K} +\vect{x}\right )\right| \leq |\cL \cap \, r\mathcal{K}| 
	\; .\]
	From the above, we see that 
	\begin{align}
	\vol\left( \bigcup_{\vect{y}\in \cL\cap \, R\mathcal{K}} \left( \vect{y}+\frac{r}{2}\mathcal{K}\right)\right) &
	    \geq |\cL \cap \, r\mathcal{K}|^{-1} \sum\limits_{\vect{y}\in \cL\cap \, R\mathcal{K}} \vol \left(\vect{y}+\frac{r}{2}\mathcal{K}\right) \nonumber\\
	    &= |\cL \cap \, r\mathcal{K}|^{-1} |\cL\cap \, R\mathcal{K}|\left(\frac{r}{2}\right)^n \vol(\mathcal{K}) 
	\label{eq1}
	\; .
	\end{align}
	On the other hand,
	\[ \bigcup_{\vect{y}\in \cL \cap \, R\mathcal{K}} \left(\vect{y}+\frac{r}{2}\mathcal{K} \right) \subset \left( R+ \frac{r}{2}\right)\mathcal{K}
	\; ,\]
	so that
	\begin{equation}\label{eq2}
	\vol \left( \bigcup_{\vect{y}\in \cL \cap \, R\mathcal{K}} \left(\vect{y}+\frac{r}{2}\mathcal{K} \right)   \right) \leq \left( R+\frac{r}{2}\right)^n \vol(\mathcal{K})
	\; .
	\end{equation}
	By Eqs.~\eqref{eq1} and~\eqref{eq2}, we get 
	\[ \left( R+ \frac{r}{2}\right)^n \geq |\cL\cap \, r\mathcal{K}|^{-1} |\cL \cap \, R\mathcal{K}| \left(\frac{r}{2}\right)^n
	\; ,\]
	and the result follows.
\end{proof}

Using the above lemma, we derive the following proposition, which will help us find a radius $r^\dagger$ such that ``$N_p(\lat, r)$ does not grow too quickly for $r \approx r^\dagger$.''

\begin{prop}\label{prop:smaller_ratio}
    For any lattice $\cL\subset \real^m$,  $c \ge 2$, radius $r>0$ and symmetric convex body $\mathcal{K}\subset \real^m$ there exists $c^\dagger$ such that $c \geq c^\dagger> c/2$ and 
    \[\frac{|\cL \cap \; c^\dagger r\mathcal{K}| }{|\cL \cap\; (c^\dagger-1)r\mathcal{K}|}\leq 2^{\epsilon m}
    \; ,
    \]
where $\epsilon=\frac{\log 5}{\lfloor c/2 \rfloor}$.
\end{prop}
\begin{proof}
Let $c' = c - \lfloor \frac{c}{2} \rfloor$. By Lemma \ref{lem:counting}, we know that 
\[ \frac{|\cL \cap \, c r\mathcal{K}|}{|\cL \cap \, c' r\mathcal{K}|} \leq 5^m.\]

Let $N_i=|\cL \cap \, (c'+i)r\mathcal{K}|$ and we get 
\[ \frac{N_1}{N_0}\cdot \frac{N_2}{N_1}\cdot \hdots \cdot \frac{N_{\lfloor \frac{c}{2} \rfloor}}{N_{\lfloor \frac{c}{2} \rfloor-1}}= \frac{N_{\lfloor \frac{c}{2} \rfloor}}{N_0} =  \frac{|\cL \cap \, c r\mathcal{K}|}{|\cL \cap \, c' r\mathcal{K}|}\leq 5^m . \]
By the pigeonhole principle there must exists a $1\leq j\leq \lfloor \frac{c}{2} \rfloor$, for which \[\frac{N_j}{N_{j-1}} \leq 5^{m/\lfloor c/2 \rfloor} = 2^{\frac{m\log 5}{\lfloor c/2 \rfloor}}\;. 
\qedhere\]
\end{proof}

Finally, we will need the following rather technical lemma.

\begin{lemma}\label{lem:bound_on_T}
 For any $\eps > 0$, $p \geq 1$, lattice $\cL=\cL(\basis)\subset \real^m$, vectors $\vect{t}, \vec{x} \in \real^m$, $\vec{v} \in \lat$, and radius $r>0$. Let $\cS \subseteq \cL \cap \, (\vec{x} + r\cB_p^m)$. Then the set $\cT = \cS + \vec{v} \subseteq \lat$ satisfies the following:
 \[
\cT - \cS - \vec{t} \subseteq \left(2r + \|\vect{v} - \vec{t}\|_p\right)\cB_p^m \;,
 \]
 and if $\vec{x} = \vec{0}$, then
 \[
 \cT - \vec{t} \subseteq \left(r + \|\vect{v} - \vec{t}\|_p\right)\cB_p^m \;.
 \]
 \end{lemma}
 \begin{proof}
 Clearly, any vector in $\cT - \cS - \vec{t}$ is contained in $\vec{v} - \vec{t} + 2r \cB_p^m$, which in turn is contained in $\left(2r + \|\vect{v} - \vec{t}\|_p\right)\cB_p^m$.
 
 Also, any vector in $\cT - \vec{t} $ is contained in $\vec{v} - \vec{t} + r \cB_p^m$, which in turn is contained in $\left(r + \|\vect{v} - \vec{t}\|_p\right)\cB_p^m$.
 \end{proof}

\begin{theorem}\label{thm:CVP_q-to-usvp_p}
	For any $\epsilon=\eps(n,m)\in (0,1/100)$, $\delta < \eps/40$, and  $1\leq p\leq q\leq \infty$,  there is a $2^{\eps m}$-time reduction from  $(2,\gamma)$-$\BDD_q$ to  $(1+\delta)$-$\uSVP_p$ where 
	\[\gamma := \frac{80}{\eps}\left(\frac{10}{\eps}\log (1/\eps)\right)^{1/p}.\]
	Furthermore, for the special case of $p = q$, we can take $\gamma = 80/\eps$.
	
	The reduction calls its $\uSVP_p$ oracle on lattices with dimension $m+1$ and rank $n+1$.
\end{theorem}
\begin{proof}
Let $\kappa:=m^{\left(\frac{1}{p}-\frac{1}{q}\right)}$.  We show a reduction that runs in polynomial time and output a close vector with probability $ \frac{2^{-\eps m}}{\poly(m \ell/\eps)}$, where $\ell$ is the bit length of the input. The result follows by running the algorithm $\poly(m \ell/\eps)\cdot 2^{\eps m}$ times.  The reduction takes as input a basis $\basis \in \real^{m \times n}$ of a lattice $\lat \subset \real^m$, and a target $\vec{t}\in \real^m$ such that $\dist_q(\vec{t}, \lat) < 2\lambda_1^{(q)}(\lat)$. By rescaling randomly with an appropriate distribution, we may assume without loss of generality that $\frac{m}{\kappa}(1-1/m)< \dist_q(\vect{t},\cL)\leq \frac{m}{\kappa}$.

Let
\[
\mathbf{B}^\dagger :=\begin{pmatrix}  \mathbf{B} & - \vect{t} \\
\vect{0} & 1
\end{pmatrix} \in \R^{(m+1) \times (n+1)}
\; , 
\]
 $\cL^\dagger:=\cL(\mathbf{B}^\dagger)$ and $\ell$ be the number of bits required for the natural representation of $\basis^\dagger$.
It follows from Proposition~\ref{prop:smaller_ratio} that by sampling $r$ appropriately, the algorithm  can guess a radius $r\leq (40/\eps)\cdot (m+1)$ such that
\begin{equation}\label{eq:condition_on_r}
\frac{N_p({\cL^\dagger},r)}{ N_p({\cL^\dagger},r-2(m+1))} \leq 2^{\epsilon m/4}
\end{equation}
with probability at least $1/\poly(m)$.
Similarly, it follows from Lemma \ref{lem:bound_on_number_of_vectors} that with probability $1/\poly(m,\ell)$, the algorithm can guess an integer $N$ that satisfies
\begin{equation}\label{eq:condition_on_N}
N\leq N_p(\cL^{\dagger,prim},r-(m+1))\leq mN.    
\end{equation}

The reduction does the following. It first uses its $(1+\delta)\text{-}\uSVP_p$ oracle to run the procedure from Theorem~\ref{thm:uniform_samplling_by_uSVP} \emph{twice}, with input lattice $\cL^\dagger$, radius $r$, and number $N$, receiving as output vectors $\vec{v}_1$ and $\vec{v}_2$ respectively. 
The reduction then finds integers $z_1,z_2 \in \Z$ that minimize $\|z_1 \vec{v}_1 - z_2 \vec{v}_2\|_q$, subject to the constraint that $z_1 \vec{v}_1 - z_2 \vec{v}_2 = \transpose{\left(\transpose{(\vect{v}-\vect{t})}, 1\right)}$ for some $\vec{v} \in \lat$. (Finding such $z_1$ and $z_2$ actually corresponds to solving a $\CVP_q$ instance over the rank-one lattice $\{z_1 \vec{v}_1 - z_2 \vec{v} \ : \ z_1 \vec{v}_1 - z_2 \vec{v}_2 = (\vec{v},0) \}$. So, this can be done efficiently. (If no $z_1,z_2$ that satisfy the constraint exist, then the algorithm simply fails.) It then outputs $\vec{v}$.

 The reduction clearly runs in time polynomial in the size of the input.
 We now show that it outputs a vector at the desired distance from the target in the $\ell_q$ norm with probability at least $\frac{1}{2^{\eps m} \cdot \poly(\ell,m/\eps)}$.  Specifically, we will find sets $\cS, \cT \subset \lat^\dagger$ such that (1) the difference between any vector in $\cT$ and any vector in $\cS$ is a vector of the form $(\vec{v} - \vec{t}, 1)$ with small $\ell_q$ norm; and (2) with the claimed probability, there exist integers $z_1,z_2$ such that $z_1\vec{v}_1 \in \cS$ and $z_2 \vec{v}_2 \in \cT$.
Let $r':=r-2(m+1)$. For an integer $k$, let 
\[
    \lat_k^\dagger := \{(\vec{v}- k \vec{t}, k) \in \R^{m+1} \ : \ \vec{v} \in \lat\} 
\]
be ``the $k$th layer of $\lat^\dagger$.'' Notice that $\lat^\dagger = \bigcup \lat_k^\dagger$, and that 
\[
\lat^\dagger \cap \, r'\cB_p^{m+1} = \bigcup_{|k| \le r'} (\lat_k^\dagger \cap \, r' \cB_p^{m+1})
\; .
\]
Therefore, there must exist a $-r' \leq k^* \leq r'$ such that $N_p(\lat_{k^*}^\dagger, r') \geq N_p(\lat^\dagger, r')/(2r'+1)$. By symmetry, we may assume that $k^* \geq 0$.

Let $\alpha=\left(\frac{10}{\eps}\log(1/\eps)\right)^{1/p}$. From Lemma \ref{lem:packing_all_norms}, we know that  $\lat_{k^*}^\dagger \cap \, r'\mathcal{B}_p^{m+1}$ can be covered by $2^{\eps m/4}$ translated copies of the $\ell_q$ ball with radius
\[
    \frac{\alpha}{\kappa} \cdot r'= \frac{\alpha}{\kappa}\cdot (r-2(m+1)) < \gamma/2 \cdot \dist_q(\vec{t}, \lat)- \left(\frac{2m}{\kappa}\right)
\]
Hence, there exists a vector $\vec{x}\in \real^{m+1}$
such that the set
\[
{\mathcal S} := \lat_{k^*}^\dagger  \cap \, \Big( \vec{x} + \frac{\alpha r'}{\kappa} {\mathcal B}_q^{m+1}\Big)  
\;, \]
  has size at least $\frac{N_p(\lat^\dagger, r')}{2^{\eps m/4}(2r'+1 )}$. 
 Let $\vect{y}\in \lat$ be a closest vector to $\vec{t}$ in the $\ell_q$ norm, i.e., $\|\vect{y}-\vect{t}\|_q=\dist_q(\vect{t},\lat)$. We define ${\mathcal T} \subset \lat_{k^*+1}^\dagger$ to be the shift of ${\mathcal S}$ by the vector $\transpose{\left(\transpose{(\vect{y}-\vect{t})}, 1\right)}$, i.e.,
 \[
 \mathcal{T} :=  \mathcal{S} + \transpose{\left(\transpose{(\vect{y}-\vect{t})}, 1\right)} \subset \lat_{k^*+1}^\dagger \;.
 \]
 Notice that $\|\transpose{(\transpose{(\vect{y}-\vect{t})}, 1)}\|_p \leq \kappa \dist_q(\vec{t},\lat)+1$ and 
 \begin{equation}\label{eq:T-S_is_in_1st_hyperplane}
      \cT-\cS \subseteq \lat_1^\dagger \; .
 \end{equation}
 From Lemma \ref{lem:bound_on_T}, we see that 
 \begin{equation}\label{eq:bound_on_T-S}
 \cT-\cS \subseteq (2\frac{\alpha}{\kappa} r'+\frac{m}{\kappa}+1)\mathcal{B}_q^{m+1}\; , \text{ and } \cT\subseteq (r'+m+1)\mathcal{B}_p^{m+1}.
\end{equation}
 
Let $\cS':=\{\vec{v}\in \cL^{\dagger,prim} : \exists k\in \intg_{> 0}, k\vec{v}\in \cS\}$ and $\cT':=\{\vec{v}\in \cL^{\dagger,prim} : \exists k\in \intg_{> 0}, k\vec{v}\in \cT\}$. By the definition of $\BDD$, we have $\lambda_1^{(p)}(\lat) \ge \lambda_1^{(q)}(\lat) \ge \frac{1}{2}\dist_q(\vec{t},\cL) >  \frac{m}{4\kappa}$. Therefore, 
\begin{equation}\label{eq:bound_on_lambda}
\lambda_1^{(p)}(\cL^\dagger) \geq \min\{1,\lambda_1^{(p)}(\cL)\}\geq \min\Big\{1,\frac{m}{4\kappa}\Big\} \geq \frac{1}{4}.
\end{equation}
Applying Claim~\ref{clm:bound_on_set_S'} then gives
\begin{equation}\label{eq:bound_on_T'_S'}
 \min\{|\mathcal{T'}|,|\mathcal{S'}|\}\geq \frac{|\mathcal{S}|}{4r} \ge \frac{N_p(\cL^\dagger, r')}{\poly(m/\eps) \cdot 2^{\eps m/4}} \;.   
\end{equation}
  
 Let $f(m):=m\frac{N_p(\cL^{\dagger,prim},r)}{N_p(\cL^{\dagger,prim},r-m-1)}$.
Notice that $\cS',\cT'\in (r-m-1)\mathcal{B}_p^{m+1}$, $\frac{r}{r-m-1}\geq 1+\delta$, $f(m)\cdot N\leq 2^{\poly(\ell,m/\eps)}$, and $N_p(\cL^{\dagger,prim},r)/f(m)\leq N \leq N_p(\cL^{\dagger,prim},r-m-1)$. We have 
\begin{align*}
 \Pr[\vect{v}_1\in \mathcal{T'}] &\geq \frac{|\cT'| \cdot N_p(\cL^{\dagger,prim},r-m-1) }{\poly(\ell,m/\eps)\cdot N \cdot N_p(\cL^{\dagger,prim},r)} &\text{(Theorem~\ref{thm:uniform_samplling_by_uSVP})} \\
 & \geq \frac{|\cT'|  }{\poly(\ell,m/\eps) \cdot N_p(\cL^{\dagger,prim},r)} &\text{(Eq.~\eqref{eq:condition_on_N})} \\
 & \geq \frac{N_p(\cL^\dagger,r')}{\poly(\ell,m/\eps)\cdot 2^{\eps m/4} \cdot N_p(\cL^{\dagger,prim},r)} &\text{(Eq.~\eqref{eq:bound_on_T'_S'})}\\
 & \geq \frac{2^{-\eps m/2}}{\poly(\ell,m/\eps)} \; .
 &\text{(Eq.~\eqref{eq:condition_on_r})}
\end{align*}

By an identical argument,
\[ \Pr[\vect{v}_2 \in \cS'] \geq \frac{2^{-\eps m/2}}{\poly(\ell,m/\eps)} \; .\]
And, since $\vec{v}_1,\vec{v}_2$ are independent, it follows that
\[
\Pr[\exists z_1, z_2 \ : \ z_1\vec{v}_1\in \cT \text{ and } z_2\vec{v}_2\in\cS] \geq \Pr[\vec{v}_1 \in \cT' \text{ and } \vec{v}_2 \in \cS'] \geq \poly(\ell,m/\eps)^{-1}\cdot 2^{-\eps m}
\; .
\]
Finally, by Eqs.~\eqref{eq:T-S_is_in_1st_hyperplane} and~\eqref{eq:bound_on_T-S}, we see that when this event happens, the reduction will output $\vec{v} \in \lat$ with 
\[
    \|\vec{v} - \vec{t}\|_q \leq 2r'\alpha/\kappa + m/\kappa + 1 \leq \gamma \dist_q(\vec{t}, \lat)
    \; ,
\]
as needed.

For the special case of $p = q$, notice that we can trivially take $\alpha = 1$ (since we do not need to scale the $\ell_p$ ball at all in order to cover another $\ell_p$ ball). The ``furthermore'' then follows immediately.
\end{proof}

\begin{corollary}\label{cor:cvpq_to_usvp_p}
For any $\epsilon = \eps(n,m) \in (0,1/100)$, $\delta<\eps/40$, and $1\leq p \leq q \le \infty$, there is a $2^{\eps m}$-time reduction from $\gamma$-$\cvp_q$ to $(1+\delta)$-$\uSVP_p$, where 
\[\gamma = \frac{120}{\eps}\left(\frac{10}{\eps}\log (1/\eps)\right)^{1/p}.\]
Furthermore, for the special case of $p = q$, we have
\[
    \gamma = 120/\eps
    \; .
\]

The reduction calls its $\uSVP_p$ oracle on lattices with dimension $m+1$ and rank $n+1$.
\end{corollary}
\begin{proof}
 This follows directly from Theorem~\ref{theorem:CVPtoCVPP} and Theorem~\ref{thm:CVP_q-to-usvp_p}.
\end{proof}

\subsection{A reduction from CVP to CVP}

We next show a reduction between $\CVP$ in different norms. To that end, we will first need the following simple lemma. For $p \geq 1$, let $D^{(p)}$ be the probability distribution over $\real^m$ whose probability density function is $e^{-C_p \|\vect{x}\|_p^p}$, where $C_p := 2^p \Gamma(1+1/p)^p = \Theta(2^p)$ is such that this yields a probability distribution (i.e., $\int_{\real^m} e^{-C_p \|\vect{x}\|_p^p} {\rm d} \vect{x} = 1$). We also define $D^{(\infty)}$ as the uniform distribution over $[-1,1]^m$.

\begin{lemma} \label{lem:continous_supergaussian}
    For any $1 \leq p < \infty$ and $1 \leq q \leq \infty$, $\vect{y} \in \real^m$, and $\eps \in (0,1/100)$, 
    \[
        \Pr_{\vect{X} \sim D^{(p)}}[\|\vect{X} - \vect{y}\|_q \leq m^{1/q} r] \geq e^{-\eps m - C_p \|\vect{y}\|_p^p}
        \; ,
    \]
    where $r := 10 \log^{1/p}(1/\eps)/C_p^{1/p}$. In particular, if $\|\vect{y}\|_p < (\eps m/C_p)^{1/p}$, then this probability is at least $e^{-2\eps m}$.
\end{lemma}
\begin{proof}
Notice that $D^{(p)}$ is a product distribution. Therefore, we have
\[
    \Pr[\|\vec{X} - \vec{y}\|_q \leq r m^{1/q}] \geq \prod_{i=1}^m \Pr[|X_i - y_i| \leq r] 
    \; .
\]
If $|y_i| > r/2 > 1$, then
\[
	\Pr[|X_i - y_i| \leq r] \geq e^{-C_p |y_i|^p}
	\; .
\]
Otherwise,
\begin{align*}
	\Pr[|X_i - y_i| \leq r] 
	&\geq \Pr[-r/2 \leq X_i \leq r/2] \\
	&= 1-2\int_{r/2}^{\infty} e^{-C_p x^p} {\rm d}x \\
	&\geq 1-2^p/r^{p-1} \cdot \int_{r/2}^{\infty} x^{p-1}e^{-C_p x^p} {\rm d}x \\
	&= 1-\frac{2^p}{p C_p r^{p-1}}e^{-C_p r^p/2^p}\\
	&\geq e^{-\eps}
	\; .
\end{align*}
The result follows.
\end{proof}

\begin{theorem}\label{thm:cvp_q<p}
For any $\epsilon = \eps(n,m) \in (0,1/100)$, $1\leq p < q \le \infty$, and $\gamma  = \gamma(n,m) \geq 1$, there is a dimension- and rank-preserving $2^{\eps m}$-time reduction from $\gamma'$-$\cvp_p$ reduces to $\gamma$-$\cvp_q$, where 
\[
\gamma' := \frac{100 \log^{1/p}(1/\eps)}{\eps^{1/p}} \cdot \gamma
\; .
\]
\end{theorem}
\begin{proof}
We show a polynomial-time reduction that succeeds with probability at least $2^{-\eps m}$. The result follows by repeating this $2^{\eps m} \cdot \poly(m)$ times and taking the output vector that is closest to the target.
 The reduction takes as input (a basis $\basis \in \R^{m \times n}$ for) a lattice $\lat \subset \R^m$. By guessing $\dist(\vec{t}, \lat)$ and rescaling, we can assume without loss of generality that $(1 - 1/m)\frac{(\eps m)^{1/p}}{2 C_p^{1/p}} \leq \dist_p(\vect{t},\cL) \leq \frac{(\eps m)^{1/p}}{2 C_p^{1/p}}$.
The reduction samples $\vect{x}\sim D^{(p)}$ (where $D^{(p)}$ is the continuous supergaussian distribution, as defined above Lemma~\ref{lem:continous_supergaussian}) and calls its $\gamma$-$\CVP_q$ oracle with lattice $\cL$ and target vector $\vec{t} := \vect{t}+\vect{x}$, receiving as output $\vec{v} \in \lat$. It simply outputs $\vec{v}$.

 It is clear that the running time of the reduction is as claimed. Let $\vect{z} \in \lat$ be a lattice vector such that $ \|\vect{z}-\vec{t}\|_p = \dist_p(\vect{t},\cL)$.
 By Lemma~\ref{lem:continous_supergaussian},
 with probability at least $2^{-\eps m}$, we will have
 \begin{equation} 
 \label{eq:x_close}
 \|\vec{t} + \vec{x} - \vec{z}\|_q \leq r
 \; ,
 \end{equation}
 where
 \[
    r := \frac{10 \log^{1/p}(10/\eps)}{C_p^{1/p}} \leq \frac{25\log^{1/p}(1/\eps)}{(\eps m)^{1/p}} \cdot \dist_p(\vec{t}, \lat)
    \; .
\]
If Eq.~\eqref{eq:x_close} holds, then $\dist_q(\cL, \vect{t} + \vec{x}) \leq rm^{1/q}$ and the $\gamma$-$\CVP_q$ oracle must output a vector $\vect{v}$ such that $\|\vect{v} - (\vect{t}+\vec{x})\|_q \leq \gamma rm^{1/q}$.
Therefore, by triangle inequality and the fact that $m^{1/q}\|\vec{y}\|_q \leq m^{1/p} \|\vec{y}\|_p$ for all $\vec{y}$, we see that Eq.~\eqref{eq:x_close} implies that 
\begin{align*}
    \|\vec{v} - \vec{t}\|_p 
        &\leq \|\vec{v} - (\vec{t} + \vec{x})\|_p +  \|\vec{t} + \vec{x} - \vec{z}\|_p + \|\vec{z} - \vec{t}\|_p \\
        &\leq \gamma m^{1/p} r + m^{1/p}r +\dist_p(\vec{t}, \lat)\\
        &\leq \Big(1+ \frac{25\log^{1/p}(1/\eps)}{\eps^{1/p}} \cdot (\gamma + 1)\Big) \cdot \dist_p(\vec{t}, \lat)\\
        &\leq \gamma' \dist_p(\vec{t}, \lat)
        \; .
\end{align*}
In other words, if Eq.~\eqref{eq:x_close} holds, then the reduction will succeed.
\end{proof}

\begin{corollary}
\label{cor:cvp_p_to_uSVP_q}
For any $\epsilon = \eps(n,m) \in (0,1/100)$, $\delta<\eps/40$, and $1\leq p \leq q \le \infty$, there is a $2^{\eps m}$-time reduction from $\gamma$-$\cvp_p$ to $(1+\delta)$-$\uSVP_q$, where 
\[\gamma := \frac{200^2 \log^{1/p}(1/\eps)}{\eps^{1+1/p}}.\]

The reduction calls its $\uSVP_q$ oracle on lattices with dimension $m+1$ and rank $n+1$.
\end{corollary}
\begin{proof}
 This follows directly from Theorem~\ref{thm:cvp_q<p} and Corollary~\ref{cor:cvpq_to_usvp_p}.
\end{proof}
\section{A better reduction from CVP to SVP for \texorpdfstring{$1\leq p\leq 2$}{p between 1 and 2}}
\label{sec:supergaussian}

In this section, we present a reduction from $\gamma$-$\cvp_p$ to $\svp_p$ for $1\leq p \leq 2$ that achieves better parameters than those implied by Corollary~\ref{cor:cvpq_to_usvp_p}.

\begin{theorem}
\label{thm:CVP_to_SVP_supergaussian}
	For any $1\leq p\leq 2$ and $\gamma \geq 4$, there is a $\cO(e^{\frac{m}{(\gamma /4)^p}})$-time reduction from $\gamma$-$\cvp_p$ over a lattice with dimension $m$ and rank $n$ to $\svp_p$ oracle over lattices with dimension $m+1$ and rank $n+1$.
\end{theorem}
Note that this also gives a polynomial time Turing reduction from $\sqrt{n/\log n}$-$\cvp_2$ to $\svp_2$ which is an improvement over Kannan's celebrated reduction from $\sqrt{n}$-$\CVP_2$~\cite{Kannan87}.

\subsection{A reduction from CVP to DSS}
We present a reduction from approximation of $\cvp$ to sampling from the discrete supergaussian distribution.~\cite{ADRS15} showed a similar result for the special case when $p = 2$ and $M = 2^{n/2}$.

\begin{theorem}
	For any $1\leq p\leq 2$, $\gamma = \gamma(n,m) \geq 4$, and $\delta>0 $, there is a $M \cdot \poly(m)$-time reduction from $\gamma$-$\cvp_p$ to $\delta$-$\dss_p^M$
	where $M=\mathcal{O}(e^{\frac{m}{(\gamma /4)^p} + \delta})$.
	
	The reduction calls its $\dss$ oracle on lattices with dimension $m +1$ and rank $n+1$.
\end{theorem}
\begin{proof}
    Given a basis $\basis \in \R^{m \times n}$ for lattice $\cL \subset \R^m$ and target vector $\vect{t}\in \real^m$ as input, the reduction behaves as follows. By randomly rescaling appropriately, we may assume that $m^{1/p}(1-1/m) < \dist_p(\vect{t},\cL)\leq m^{1/p}$. 

    Let $\alpha = \gamma /4$ and $\cL^\dagger$ be the lattice generated by basis
	\[ \basis^\dagger=\begin{bmatrix} \frac{1 }{\alpha}\basis & -\frac{1}{\alpha}\vect{t} \\
	\vect{0} & 1
	\end{bmatrix}.\]
	The reduction then does the following.
	\begin{enumerate}
	    \item It uses its oracle to sample $100\cdot e^{m/\alpha^p + \delta}$ vectors from a distribution similar to $D_{\lat^\dagger, p}$.
	    \item Let $\vect{v}$ be the shortest vector among the returned vector whose last coordinate is $1$. The reduction outputs the first $m$ coordinates of $\alpha\vect{v}-(-\vect{t},\alpha)$.
	\end{enumerate}

	   The running time is clearly as claimed. Consider the set $\mathcal{G}$ consisting of vectors from $\cL^\dagger$ whose coefficient on the last basis vector $(-\vect{t}/\alpha, 1)$ is exactly 1: 
	\[\mathcal{G}:=\{\vect{y}\in \cL^\dagger \;| \; \; \langle \vect{y},\vect{e_{m+1}}\rangle =1\} \; .\]	
	Let $\vect{v}$ be a random vector sampled from  distribution $D_{\cL^\dagger,p}$.
	We want to lower bound the probability of $\vect{v} \in \mathcal{G}$ since $\mathcal{G}$ consists of vectors that we are interested in. We start by upper bounding $f_p(\cL^\dagger)$.
	\begin{align*} 
		f_p(\cL^\dagger) &= \sum\limits_{k=-\infty}^{\infty} \sum_{\vect{u} \in \cL} f_{p}\left(\frac{\vect{u}-k\vect{t}}{\alpha}, k\right)\\
		&=  \sum\limits_{k=-\infty}^{\infty} f_{p}(k) \sum_{\vect{u} \in \cL} f_{p}\left(\frac{\vect{u}-k\vect{t}}{\alpha} \right)\\
		&\leq \sum\limits_{k=-\infty}^{\infty} f_{p}(k) \cdot f_{p}\left(\frac{1}{\alpha}\cL\right)\\
		&\leq c \cdot f_{p}\left(\frac{1}{\alpha}\cL\right) \; ,
	\end{align*}
	where $c>0$ is a constant. Notice that the second-to-last inequality follows from Lemma~\ref{lem:shifted_supergaussian_mass}. Let $\vect{y}\in \cL$ be a vector such that $\|\vect{y}-\vect{t}\|_p=\dist_p(\vect{t},\cL)$. On the other hand, we can lower bound $f_p(\mathcal{G})$:
	\begin{align*}
		f_p(\mathcal{G}) &= \sum_{\vect{u}\in \cL}f_{p}\left(\frac{1}{\alpha}\vect{u}+\frac{1}{\alpha}(\vect{y}-\vect{t}), 1   \right) \\
		&\geq f_p\left(\frac{1}{\alpha}(\vect{y}-\vect{t}),1\right) \cdot \sum_{\vect{u}\in \cL} f_p\left(\frac{1}{\alpha}\vect{u}\right) 
		\\
		&\geq e^{-1 - \frac{\|\vect{y}-\vect{t}\|_p^p}{\alpha^p}}\cdot f_p\left( \frac{1}{\alpha}\cL \right) \\
		&\geq e^{-1 - m/\alpha^p}\cdot f_p\left(\frac{1}{\alpha}\cL\right) \\
		&\geq \frac{1}{3}e^{-m/\alpha^p}f_p\left(\frac{1}{\alpha}\cL\right) \; ,
	\end{align*}
	where the first inequality follows from Lemma~\ref{lem:shifted_supergaussian_mass}.
	Therefore, 
	\[\Pr[\vect{x}\in \mathcal{G}] = \frac{f_{p}(\mathcal{G})}{f_{p}(\cL^\dagger)} \geq \frac{1}{3c} e^{-m/\alpha^p} \; .\]
    Next, by Lemma \ref{lem:suppergaussian}, 
	 \[ \Pr[\|\vect{v}\|_p \geq t m^{1/p} ] \leq (e^{-t^p+1/p}\cdot t\cdot p^{1/p})^m \; .\]  
	For $t\geq 4$ we know that 
	\begin{equation}
	e^{-t^p+1/p}\cdot t \cdot p^{1/p} \leq e^{-1}	\; ,
	\end{equation}
    so that the above probability is at most $e^{-m}$.
    Then, by union bound, we have
	\[\Pr[\vect{v}\in \mathcal{G}\, \& \,\|\alpha \vect{v}\|_p \leq 4 \alpha\cdot \dist_p(\vect{t},\cL) ] \geq 1 - \left(1 - \frac{1}{3c}e^{-m/\alpha^p} + e^{-m}\right) = \Omega(e^{-m/\alpha^p}) \; .\]
	
	It follows that each vector output by the $\delta$-$\dss_p^M$ oracle will yield a correct solution with probability at least $e^{-m/\alpha^p - \delta}$, and the result follows.
\end{proof}

\cite{Ste16} gave a polynomial time reduction from $\sqrt{n/\log n}$-$\svp_2$ to $\dss_2$. As a corollary of the above theorem we also get a similar reduction for $\cvp_2$.
\begin{corollary}
There is a polynomial-time rank-preserving reduction from $\sqrt{n/\log n}$-$\cvp_2$ to $\delta$-$\DSS_2$ for any $\delta \leq O(\log n)$.
\end{corollary}

\subsection{A reduction from DSS to SVP}

The following result gives a reduction from DSS to SVP. The reduction was shown in~\cite{Ste16} for $p = 2$. The proof for general $p$ is nearly identical. We include it in Appendix~\ref{app:proof_of_DSS} for completeness.
\begin{restatable}
[Theorem 4.6, \cite{Ste16}]
{theorem}{DSStoSVP}\label{thm:DSS_to_SVP}
For any efficiently computable function $f(m)$ with $1\leq f(m)\leq \poly(m)$, there is an (expected) polynomial time reduction from $\delta$-$\DSS_p$ to $\SVP_p$ where $\delta=\frac{1}{2}\ln{(1+1/f(m))}$.
The reduction preserves dimension and rank and only calls the $\SVP_p$ oracle on sub-lattices of the input lattice.
\end{restatable}

\section*{Acknowledgments} We would like to thank Moritz Venzin for suggesting a simplification in Theorem~\ref{thm:svp_to_svp}. We would also like to thank the SODA reviewers for their detailed and helpful reviews. 

Research at CQT is funded by the National Research Foundation, the Prime Minister's Office, and the Ministry of Education, Singapore under the Research Centres of Excellence programme's research grant R-710-000-012-135. The third named author was also supported in part by the National Research Foundation Singapore under its AI Singapore Programme [Award Number: AISG-RP-2018-005].

The last named author was partially supported by the Simons Institute in Berkeley.

\appendix
\section{Additional preliminaries}
\subsection{Proof of Theorem~\ref{thm:uniform_samplling_by_uSVP}}

We will need the following slight variant of the Lemma \ref{lem:almostindependent}.

\begin{lemma}
\label{lem:some_shifted}
	For any prime $Q$ and vectors $ \vec{v}_1,\ldots, \vec{v}_N, \vec{y}_1,\ldots, \vec{y}_N \in \Z_Q^n$ with $\vec{v}_i \neq \vec0 \bmod Q$ and $\vec{y}_i \neq \vec{y}_1$,
	\[
		\frac{1}{Q} - \frac{2N}{Q^2} - \frac{N}{Q^n} \leq \Pr[\langle \vec{z}, \vec{y}_1 +\vec{c} \rangle = 0 \text{ and } \langle \vec{z}, \vec{v}_i \rangle \neq 0 \ \forall i \text{ and } \langle \vec{z}, \vec{y}_i + \vec{c} \rangle  \neq 0 \ \forall i > 1] \leq \frac{1}{Q} + \frac{1}{Q^n}
		\; ,
	\]
	where $\vec{z}, \vec{c} \in \Z_Q^n$ are sampled uniformly and independently at random, and all inner products are modulo $Q$.
\end{lemma}
\begin{proof}
	The upper bound follows from the observation that 
	\[
		\Pr[\langle \vec{z}, \vec{y}_1 +\vec{c} \rangle = 0 \bmod Q] \leq \Pr[\vec{y}_1+ \vec{c} = \vec0 \bmod Q] + \Pr[\langle \vec{z}, \vec{y}_1+ \vec{c} \rangle = \vec0 \bmod Q \ | \ \vec{y}_1 + \vec{c} \neq \vec0] = \frac{1}{Q^n} + \frac{1}{Q}
		\; .
	\]
	For the lower bound, it suffices to show that for any fixed $i$,
	\[
		\Pr[\langle \vec{z}, \vec{v}_i \rangle = 0 \bmod Q \ | \ \langle \vec{z}, \vec{y}_1 +\vec{c} \rangle = 0 \bmod Q] \leq \frac{1}{Q} + \frac{1}{Q^n}
		\; ,
	\]
	and that for $i>1$,
	\[
	    \Pr[\langle \vec{z}, \vec{y}_i + \vec{c} \rangle = 0 \bmod Q \ | \ \langle \vec{z}, \vec{y}_1 +\vec{c} \rangle = 0 \bmod Q] \leq \frac{1}{Q} + \frac{1}{Q^n}
	    \; .
	\]
	The result then immediately follows from the observation that $\Pr[\langle \vec{z}, \vec{y}_1 +\vec{c} \rangle = 0 \bmod Q] \geq 1/Q$ and union bound.
	
	Indeed, the statement for the $\vec{y}_i$ is immediate from the observation that for any fixed $i>1$ the random variables $\langle \vec{z}, \vec{y}_i + \vec{c} \rangle \bmod Q$ is uniformly random and independent of $\langle \vec{z}, \vec{y}_1 + \vec{c} \rangle \bmod Q$ (since $\vec{y}_i \neq \vec{y}_1$).
	For $\vec{v}_i$, suppose that $\vec{c}' \in \Z_q^n$ is such that $\vec{y}_1 + \vec{c}' \neq \vec{v}_i \bmod Q$. Then,
	the random variables $\langle \vec{z}, \vec{v}_i - \vec{y}_1 - \vec{c}' \rangle $ is uniformly random modulo $Q$ and independent of $\langle \vec{z}, \vec{y}_1 + \vec{c}' \rangle$. Therefore,
	\begin{align*}
		\Pr[\langle \vec{z}, \vec{v}_i \rangle = 0  \ | \ \langle \vec{z}, \vec{y}_1 +\vec{c} \rangle = 0 ] 
			&\leq \Pr[\langle \vec{z}, \vec{v}_i \rangle = 0 \ | \ \langle \vec{z}, \vec{y}_1 +\vec{c} \rangle = 0,\ \vec{y}_1 + \vec{c} \neq \vec{v}_i] + 1/Q^n \\
			&= 1/Q + 1/Q^n
		\; ,
	\end{align*}
	as needed.
\end{proof}

The following lemma is an $\ell_p$ generalization of~\cite[Lemma 2.18]{Ste16}. The proof is identical, but we include it for completeness.

\begin{lemma}
\label{lem:nogoodnameforthislemma}
For any $p \geq 1$, lattice $\lat \subset \R^m$ with basis $\basis\in \real^{m\times n}$, suppose $\vec{x}_1, \vec{x}_2 \in \lat$ are primitive with $\vec{x}_1 \neq \pm \vec{x}_2$ and $\|\vec{x}_1\|_p \geq \|\vec{x}_2\|_p$ such that 
\[
\basis^{-1}\vec{x}_1 = \alpha \basis^{-1}\vec{x}_2 \bmod Q
\; 
\]
for any prime number $Q \geq 100$, and $\alpha \in \Z_Q$. Then, $N_p(\cL^{prim},\|\vec{x}_1\|_p) > Q/(20\log Q)$.
\end{lemma}
\begin{proof}
We must have $\alpha \neq 0 \bmod Q$, since $\vec{x}_1$ is primitive.
So, we have that $\vec{x}_1 - q \vec{x}_2 \in Q \lat \setminus \{ \vec0\}$ for some integer $q = \alpha \bmod Q $ with $0 < |q| \leq Q/2$. Let $\vec{y} := (\vec{x}_1 - q \vec{x}_2)/Q \in \lat$ and note that $\vec{y}$ is not a multiple of $\vec{x}_2$. It suffices to find at least $\lceil Q/(20 \log Q) \rceil $ primitive vectors in the lattice spanned by $\vec{y}$ and $\vec{x}_2$ that are at least as short as $\vec{x}_1$. (Such vectors are either themselves primitive in $\lat$ or multiples of distinct primitive vectors in $\lat$.)

We consider two cases. If $q = \pm 1$, then for $i = 0, \ldots, Q-1$, the vectors $i\vec{y} + q \vec{x}_2$ are clearly primitive in the lattice spanned by $\vec{y}$ and $\vec{x}_2$, and we have
\[
\|i\vec{y} + q \vec{x}_2\|_p = \|i\vec{x}_1 + q(Q-i)\vec{x}_2\|_p/Q \leq \|\vec{x}_1\|_p
\; ,
\]
as needed.

Now, suppose $|q| > 1$. Then, for $i = \lceil Q/4 \rceil,\ldots, \lfloor Q/2 \rfloor$, let $k_i$ be an integer such that $|k_i - iq/Q| \leq 1/2$ and $0 < |k_i| < i$. (Note that such an integer exists, since $1/2 \leq |iq/Q| \leq i/2$). Then,
\begin{align*}
\|i\vec{y} + k_i\vec{x}_2\|_p &= \|i \vec{x}_1/Q + (k_i -iq/Q) \vec{x}_2\|_p \leq \|\vec{x}_1\|_p
\;.
\end{align*}
When $i$ is prime, then since $0 < |k_i| < i$, we must have $\gcd(i, k_i) = 1$. Therefore, the vector $i\vec{y} + k_i\vec{x}_2$ must be primitive in the lattice spanned by $\vec{y}$ and $\vec{x}_2$ when $i$ is prime. It follows from a suitable effective version of the Prime Number Theorem that there are at least $\lceil Q/(20 \log Q) \rceil$ primes between $\lceil Q/4 \rceil$ and $\lfloor Q/2 \rfloor$ (see, e.g., \cite{rosser41}), and the result follows.
\end{proof}

From Lemmas~\ref{lem:almostindependent} and~\ref{lem:nogoodnameforthislemma}, we immediately derive the following, which is an $\ell_p$ generalization of~\cite[Theorem 4.1]{Ste16}, and its algorithmic corollary, which generalizes~\cite[Lemma 4.3]{Ste16}.

\begin{theorem}
\label{thm:uniform_sample_vector}
For any lattice $\cL=\cL(\basis)\subset \real^m$ of rank $n$, $p\geq 1$, primitive lattice vectors $\vect{v_0},\vect{v_1},\hdots,\vect{v_N} \in \cL$ with $\vect{v}_0 \neq\pm \vect{v}_i$ for all $i>0$, prime $Q\geq 101$, if $N_p(\cL^{prim},\|\vect{v_i}\|_p) \leq \frac{Q}{20\log Q}$ for all $i\leq N$, then 
\[\frac{1}{Q}-\frac{N}{Q^2}\leq Pr[\langle \vect{z},\basis^{-1}\vect{v_0}\rangle=0 \mod Q \textsf{ and } \langle \vect{z},\basis^{-1}\vect{v_i}\rangle \neq 0 \mod Q ,\forall{i}>0 ] \leq \frac{1}{Q},\]
where $\vect{z}\in \intg_Q^n$ is chosen uniformly at random.
\end{theorem}

 \begin{corollary} \label{cor:uniform_sampling}
For any constant $C > 0$, there is an expected polynomial time algorithm with access to a $\svp_p$ oracle that takes as input $p \geq 1$, a (basis for a) lattice $\cL \subset \R^m$ of rank $n$, a radius $r>0$, and an integer $N\geq 1$ and outputs a vector $\vect{y} \in \cL$ such that, if $N\leq N_p(\cL^{prim},r)\leq n^C N$ and $\lambda_1^{(p)}(\cL)> \frac{r}{n^C N_p(\cL^{prim},r)}$ then for any $\vect{x} \in \cL^{prim} \cap \, \mathcal{B}_p(r)$,
\[\frac{1-n^{-C}}{N_p(\cL^{prim},r)} \leq Pr[\vect{y}=\pm \vect{x}]\leq \frac{1+n^{-C}}{N_p(\cL^{prim},r)}
\; .\]
The algorithm preserves rank and dimension and only calls the oracle on sublattices of the input lattice.
\end{corollary}

Finally, we can prove the theorem.

\uniformSampllingByuSVP*

\begin{proof}\label{proof:uniformSamplingByuSVP}
    On input $p \geq 1$, a basis $\basis\in \real^{m\times n}$ for a lattice $\lat \subset \R^m$,  algorithm first finds a prime $Q$ such that $100 f(m)N \log(f(m)N) \leq Q \leq 200 f(m)N \log(f(m)N)$. It then samples $\vec{z} \in \Z_Q^n$ uniformly at random and sets
    \[
        \lat' := \{\vec{v} \in \lat \ : \ \langle \vec{z}, \basis^{-1} \vec{v} \rangle = 0 \bmod Q\}
        \; ,
    \]
    and calls its $\gamma\text{-}\uSVP_p$ oracle on $\lat'$. If the oracle outputs $\vec{v} \in \lat$ such that $\|\vec{v}\|_p \leq r$, the algorithm outputs $\pm \vec{v}$ (where the sign is chosen randomly). Otherwise, it outputs $\vec0$.
    
    The running time of the algorithm is clear. Let $\vec{x}_0 \in \cL^{prim} \cap \, \mathcal{B}_p(r)$. Let $\vec{x}_1,\ldots, \vec{x}_M \in \lat$ be all distinct primitive lattice vectors with $\|\vec{x}_i\| \leq \gamma r$ and $\vec{x}_i \neq \pm \vec{x}_0$. By Theorem~\ref{thm:uniform_sample_vector}, we have
    \begin{align*}
         \Pr[\vec{x}_0 \in \lat' \text{ and } \vec{x}_i \notin \lat',\ \forall i > 0] 
            &= Pr[\langle \vect{z},\basis^{-1}\vect{v_0}\rangle=0 \bmod Q \text{ and } \langle \vect{z},\basis^{-1}\vect{v_i}\rangle \neq 0 \bmod Q ,\ \forall{i}>0 ] \\
            &\geq \frac{1}{Q}-\frac{M}{Q^2}
        \; .
    \end{align*}
    Notice that whenever this event occurs, the oracle must output $\pm \vec{x}_0$.
    By assumption, $M < N_p(\cL^{prim},\gamma(n) r) \ll Q/2$, so this probability is at least $1/(2Q) \geq 1/(1000 f(m) N \log(f(m)N))$, and the result follows.
\end{proof}

\subsection{Proof of Theorem~\ref{theorem:CVPtoCVPP}}

\CVPtoCVPP*

\begin{proof}\label{proof:CVPtoCVPP}
    The reduction takes as input a $\basis \in \real^{m\times n}$ for a lattice $\lat \subset \R^m$, and a target vector $\vec{t} \in \R^m$. If $\lambda_1^{(p)}(\lat) > \dist_p(\vec{t}, \lat)$, then the reduction is trivial. So, we may assume that $\lambda_1^{(p)}(\lat) \leq \dist_p(\vec{t}, \lat)$. Let $r := \dist_p(\vec{t}, \lat)/\tau$. We may also assume without loss of generality that the reduction takes as input a prime number $Q \geq 101$ such that and $100  \cdot N_p(\lat, r) \leq Q \leq 200 N_p(\lat, r)$, since the reduction can guess this value with probability at least $1/\poly(m)$.
    
    The reduction then samples $\vec{z} \in \Z_Q^n$ and $\vec{c} \in \Z_Q^n$ uniformly and independently at random, sets $\vec{t}' := \vec{t} + \vec{y}$ where $\vec{y}$ is any lattice vector with $\langle \vec{z}, \basis^{-1}\vec{y} - \vec{c}\rangle = 0 \bmod Q$ and 
    \[
        \lat' := \{ \vec{v} \in \lat \ : \ \langle \vec{z}, \basis^{-1} \vec{v} \rangle = 0 \bmod Q \}
        \; .
    \]

    It then calls its $(\alpha, \gamma)\text{-}\BDD_p$ oracle on input $\lat'$ and $\vec{t}'$, receiving as output $\vec{v}$. Finally, the reduction outputs $\vec{v} - \vec{y}$.
    
    The running time of the reduction is clear. For correctness, it suffices to show that with positive constant probability, we have both that $\dist_p(\vec{t}', \lat') \leq (1+1/\tau)\dist_p(\vec{t}, \lat) $ and $\dist_p(\vec{t}', \lat') \leq (1+\tau)\lambda_1^{(p)}(\lat')$. 
    
    To that end, let $\vec{v}_1,\ldots, \vec{v}_N \in \lat$ be all distinct non-zero vectors with length at most $r$, and let $\vec{y}_i := \vec{v}_i + \vec{w} \in \lat$, where $\vec{w}$ is a closest lattice vector to $\vec{t}$. Notice that $\lambda_1^{(p)}(\lat') > r$ if and only if $\vec{v}_i \notin \lat'$ for all $i$. Furthermore, if there exists a $\vec{y}_i \in \lat' + \vec{y}$, then $\dist_p(\vec{t}', \lat') \leq \|\vec{y}_i - \vec{t}\|_p \leq \dist_p(\vec{t}, \lat) + r$. If both of these events occur simultaneously with non-negligible probability, then we are done.
    
    Notice that $\vec{v}_i \notin Q\lat$, since otherwise $\vec{v}_i/Q, 2\vec{v}_i/Q,\ldots, Q\vec{v}_i/Q \in \lat \cap \, B_p^n(r)$, contradicting the assumption that $N_p(\lat, r) < Q$. Therefore, we may apply Lemma~\ref{lem:some_shifted} to $\basis^{-1}\vec{v}_i$ and $\basis^{-1} \vec{y}_i$ to see that for each $i$,
    \[
        \Pr[\vec{y}_i \in \lat + \vec{y} \text{ and }\vec{v}_j \notin \lat\ \forall j\text{, and } \vec{y}_j \notin \lat' + \vec{y}\ \forall j \neq i] \geq \frac{1}{Q} - \frac{2N}{Q^2} - \frac{N}{Q^n}
        \; .
    \]
    Notice that these are disjoint events, so that the probability of at least one happening is exactly the sum of each probability. It follows that
    \[
        \Pr[\exists i,\ \vec{y}_i \in \lat + \vec{y} \text{ and } \vec{v}_j \notin \lat'\ \forall j] \geq \frac{N}{Q} + -\frac{2N^2}{Q^2} - \frac{N^2}{Q^n} \geq \frac{1}{1000}
        \; ,
    \]
    as needed.
\end{proof}

\subsection{Proof of Theorem~\ref{thm:DSS_to_SVP}}
\label{app:proof_of_DSS}

\begin{theorem}[Theorem 4.5,~\cite{Ste16}]\label{thm:counting_primitive_vectors}
For any efficiently computable function $f(m)$ with $1\leq f(m)\leq poly(m)$, there is a polynomial-time reduction from $(\beta,\gamma)$-$\gappvcp_p$ to $\SVP_p$ where $\beta=\frac{1}{f(m)}$ and $\gamma=1+\frac{1}{f(m)}$. The reduction preserves rank and dimension and only calls the $\svp_p$ oracle on sublattices of the input lattice.
\end{theorem}

\DSStoSVP*

\begin{proof}
On input $\cL \subset \ratn^m$, the reduction behaves as follows. First it computes $\lambda_1^{(p)}(\cL)$ using its $\SVP_p$ oracle. For $i=0,\hdots,\ell=200\cdot m^2f(m)$, let 
\[r_i= \left(\left(\lambda_1^{(p)}(\cL)\right)^p+\frac{i}{100\cdot mf(m)}\right)^{1/p}.\] Let $\gamma=1+1/f(m)$. For each $i$, the reduction uses $\SVP_p$ oracle (using the procedure from Theorem \ref{thm:counting_primitive_vectors}) to compute $N_i$ such that 
\[ \gamma^{-1/10} N_p(\cL^{prim},r_i) \leq N_i \leq N_p(\cL^{prim},r_i) \]
or $N_i=1$ if $\lambda_1^{(p)}(\cL)\leq \frac{r_i}{100\cdot m^2f(m)\cdot N_p(\cL^{prim},r_i)}$.
Let $w_\ell=f_p(\intg_{\neq 0}r_\ell )$ and for $i=0,\hdots,\ell-1$, let 
$w_i=f_p(\intg_{\neq 0}r_i)-f_{p}(\intg_{\neq 0}r_{i+1})$.

Let $W=\sum_{i=0}^\ell N_i w_i$ . Then the reduction outputs $\vect{0}$ with probability $\frac{1}{1+W}$. Otherwise it choose an index $0\leq k\leq \ell$, assigning to each index $i$ probability $\frac{N_i\cdot w_i}{W}$. If $N_k>1$, the reduction chooses a vector $\vect{x}\in \cL^{prim}$ that is uniformly distributed over $\cL^{prim}\cap \, \mathcal{B}_p(r_k)$, up to a factor of $\gamma^{\pm 1/10} $. If $N_k=1$, the reduction simply sets $\vect{x}=\SVP_p(\cL)$. Finally, it samples an integer $z$ from $D_{\Z_{\neq 0} \|\vec{x}\|_p,p}$ and returns $\vect{v}=z\cdot \vect{x}$.

First, we note that the reduction runs in expected polynomial time. In particular, the $N_i$ have polynomial bit length by Lemma \ref{lem:bound_on_number_of_vectors}, and the various subprocedures have expected running times that are polynomial in the length of their input.

We now prove correctness. Let $\cL^\dagger$ be the set of all lattice vectors that are integer multiples of a primitive lattice vector whose length is at most $m^{1/p}$. By the definition, it is enough to bound the probability of all the vector from $\lat^\dagger$. Then,
\[f_p(\cL^\dagger\setminus \{\vect{0}\})=\sum\limits_{\vect{y}\in \cL^\dagger \setminus \{\vect{0}\}}f_p(\vect{y})= \sum\limits_{\vect{y}\in\cL^{prim} \cap \, m^{1/p}\mathcal{B}_p} f_p(\Z_{\neq 0} \|\vect{y}\|_p) 
\; .\]  
For any $\vect{y}$ with $r_{i-1}\leq \|\vect{y}\|_p\leq r_{i}$, we have 
\[ f_p(\Z_{\neq 0}r_i)\leq f_p(\Z_{\neq 0}\|\vect{y}\|_p) \leq \gamma^{1/10} f_p(\Z_{\neq 0}r_{i} )
\; .\]
From the definition of $w_i$, we have 
\[ \sum\limits_{i=0}^\ell w_i \cdot N_p(\cL^{prim}, r_i) \leq f_p(\cL^\dagger\setminus \{0\}) \leq \gamma^{1/10}\sum\limits_{i=0}^\ell w_i \cdot N_p(\cL^{prim},r_i)  
\; .\]

Now we would like to say that $N_i \approx N_p(\cL^{prim},r_i)$. This is true by definition except when $N_i=1$ and $N_p(\cL^{prim},r)>1$, i.e., when $\lambda_1^{(p)}(\cL)< \frac{r_i}{100m^2f(m)N_p(\cL^{prim},r_i)}$ and $\lambda_2^{(p)}(\cL)\leq r_i$. Notice that 
\[N_p(\cL^{prim},r_{i+1}) \geq \frac{r_{i+1}-\lambda_2^{(p)}(\cL)}{\lambda_1^{(p)}(\cL)}\geq \frac{1}{200m\cdot f(m)\lambda_1^{(p)}(\cL)}. \]
This implies that for $j>i$, we get $\lambda_1^{(p)}(\cL)> \frac{r_j}{100m^2\cdot f(m)N_p(\cL^{prim},r_j)}$.
 It follows that, for any $i<\ell$, we have 
 \[\gamma^{-1/5}\cdot \sum\limits_{j\geq i} w_j\cdot N_p(\cL^{prim},r_j) \leq \sum\limits_{j\geq i} w_j\cdot N_j  \leq \sum\limits_{j\geq i} w_j\cdot N_p(\cL^{prim},r_j)
 \; .\]
Therefore, we have that,
\[\gamma^{-1/5} f_p(\cL^\dagger\setminus \{\vect{0}\})\leq W\leq \gamma^{1/5} f_p(\cL^\dagger\setminus \{\vect{0}\}) \; . \]

So, the probability that reduction outputs $\vect{0}$ is $\frac{1}{1+W}$, which is a good approximation to the correct probability of $\frac{1}{f_p(\cL^\dagger)}$.

Now, for any $\vect{y}\in \cL^{prim}$, it follows from Corollary \ref{cor:uniform_sampling} that 
\begin{equation}\label{eq:bound_on_y}
\gamma^{-1/2}\frac{f_p(\intg_{\neq 0}\|\vect{y}\|_p )}{f_p(\cL^\dagger)}\leq \Pr[\vect{x}=\pm \vect{y}]\leq \gamma^{1/2}\frac{f_p(\Z_{\neq 0}\|\vect{y}\|_p)}{f_p(\cL^\dagger)} 
\; .\end{equation}

Finally, for any $\vect{u}\in \cL^\dagger\setminus \{\vect{0}\}$, Let $\vect{y}$ be one of the primitive lattice vector that are scalar multiples of $\vect{u}$ and let $z'$ be such that $\vect{u}=z'\vect{y}$. Then,
\[\Pr[\vect{v}=\vect{u}]=\Pr[\vect{x}=\pm \vect{y}]\cdot\Pr[z=z'| \vect{x}=\pm \vect{y}]\]
\[=\Pr[\vect{x}=\pm \vect{y}].\frac{f_p(\vect{u})}{f_p({\|\vect{y}\|_p\cdot \intg_{\neq 0}})}
\; .\]
The result follows from plugging the above equation into Eq.~\eqref{eq:bound_on_y}.

\end{proof}

\end{document}